\documentclass[sigconf, nonacm]{acmart}

\makeatletter
\def\@ACM@checkaffil{% Only warnings
    \if@ACM@instpresent\else
    \ClassWarningNoLine{\@classname}{No institution present for an affiliation}%
    \fi
    \if@ACM@citypresent\else
    \ClassWarningNoLine{\@classname}{No city present for an affiliation}%
    \fi
    \if@ACM@countrypresent\else
        \ClassWarningNoLine{\@classname}{No country present for an affiliation}%
    \fi
}
\makeatother

\usepackage{pvldb}% VLDB/PVLDB formatting overrides, see pvldb.sty for details.
\renewcommand\vldbdoi{XX.XX/XXX.XX}
\renewcommand\vldbpages{XXX-XXX}
\renewcommand\vldbavailabilityurl{URL_TO_YOUR_ARTIFACTS}

\usepackage{caption}
\usepackage{amsmath,amsfonts}
\usepackage{amsthm}
\usepackage{graphicx}
\usepackage{textcomp}
\usepackage{subfigure}
\usepackage{booktabs}
\usepackage{xcolor}
\usepackage{multirow}
\usepackage{array}
\usepackage{threeparttable}
\usepackage[ruled,vlined]{algorithm2e}
\usepackage{balance}
\usepackage{makecell}
\usepackage{capt-of}

\newtheorem{definition}{Definition}

\newtheorem{lemma}{Lemma}

\newtheorem{theorem}{Theorem}

\newcommand{\cbullet}{\raisebox{0.15ex}{\small$\bullet$}\hspace{0.5em}}

\begin{document}
\title[
Time-Decayed Vector Search in the Rhythm of TANGO: Jointly Modeling Semantic Similarity and Temporal Freshness
]{Time-Decayed Vector Search in the Rhythm of TANGO: \\Jointly Modeling Semantic Similarity and Temporal Freshness}

%%
%% The "author" command and its associated commands are used to define the authors and their affiliations.
\author{Jiuqi Wei}
\affiliation{%
  \institution{OceanBase, Ant Group}
  % \city{Beijing}
  % \country{China}
}
\email{weijiuqi.wjq@antgroup.com}

\author{Qiyao Luo}
\affiliation{%
  \institution{OceanBase, Ant Group}
  % \city{Shanghai}
  % \country{China}
}
\email{luoqiyao.lqy@antgroup.com}

\author{Quanqing Xu}
\affiliation{%
  \institution{OceanBase, Ant Group}
  % \city{Singapore}
  % \country{Singapore}
}
\email{xuquanqing.xqq@oceanbase.com}

\author{Chuanhui Yang}
\affiliation{%
  \institution{OceanBase, Ant Group}
  % \city{Beijing}
  % \country{China}
}
\email{rizhao.ych@oceanbase.com}

\author{Themis Palpanas}
\affiliation{%
  \institution{LIPADE, Universit{\'e} Paris Cit{\'e}}
  % \city{Paris}
  % \country{France}
}
\email{themis@mi.parisdescartes.fr}

\begin{abstract}
Vector search typically measures relevance through semantic similarity under a fixed scoring function. However, in a growing range of applications, relevance may evolve over time, making temporal freshness an additional signal beyond semantic similarity. In this paper, we formalize \emph{time-decayed vector search} (TDVS), which incorporates continuous temporal decay into the search objective so that relevance is jointly determined by semantic similarity and temporal freshness. We design \emph{Score-Preserving Temporal Reduction} (STR) that enables existing Maximum Inner Product Search indexes to directly support TDVS. We further present \emph{Chronos}, a TDVS-native framework that derives an exact metric formulation and introduces \emph{Query-Orthogonal TimeLift} to control data--data geometry while preserving all query--data scores and rankings. Building on Chronos, we propose \emph{TANGO}, a hierarchical graph index that adopts layer-specific TimeLift geometries to preserve temporal locality at the base layer while strengthening long-range semantic connectivity in upper layers. TANGO traverses the hierarchy using the exact TDVS score, caches temporal factors to reduce computation, and supports efficient online insertion. Extensive experiments show that TANGO achieves up to $3.5\times$ higher query throughput and $4.05\times$ faster index construction than state-of-the-art graph-based competitors. TANGO also maintains its advantage over all competitors across diverse temporal settings and enables efficient online insertion, demonstrating its robustness and practicality.
\end{abstract}

\maketitle

\vldbtopmatter

\section{Introduction}
%\textbf{Background and Problem.}
Vector search has become a fundamental building block of modern data management and AI systems, supporting a wide range of applications such as semantic search, retrieval-augmented generation (RAG), large language models (LLMs), recommender systems, and agentic workflows~\cite{DBLP:conf/wims/EchihabiZP20, wei2026virtuouscycleaipoweredvector}. 
In its standard form, vector search returns the items that are most similar to a query under a fixed similarity function, such as inner product, cosine similarity, or Euclidean distance~\cite{wang2023graph,azizi2025graph}. This formulation works well in many cases, but it assumes that relevance depends only on semantic similarity and remains unchanged once an item's embedding is fixed. Such an assumption no longer holds in many emerging applications, where \emph{temporal freshness} is becoming an increasingly important factor in determining relevance~\cite{ouyang2025hoh,tang2025evowiki,wu2024time}.

\textbf{Emerging Scenarios.}
Consider several representative examples. 
In agent memory~\cite{park2023generative,openclaw_memory_search}, an agent may store a long history of past interactions, observations, and intermediate reasoning results. For a new query, an old memory may be semantically very similar, yet still be less useful than a more recent one that better reflects the agent's current state, goals, or environment. 
In freshness-sensitive RAG~\cite{qian2024timer4} and temporal question answering~\cite{kasai2023realtime}, older passages may still retain partial utility, but should gradually become less competitive than more recent passages as the context evolves. 
In recency-sensitive web search~\cite{dong2010towards}, the relevance of a page may change on the scale of days or even hours for rapidly evolving events, requiring search results to account for both topical relevance and freshness.
In dynamic content recommendation~\cite{ryu2025news}, older items may still match user interests, but should often be ranked below newer items whose utility is more aligned with current attention and demand.
This demand is also reflected in the industry: Qdrant, for example, supports exponential time-based score boosting for news and other freshness-sensitive searches~\cite{qdrant_search_relevance}.

\textbf{Existing Time-Aware Pipelines.}
Existing time-aware vector-search approaches mainly follow two paradigms. The first is \emph{filter-based} search, where time is represented as metadata and range filters are applied before or during search~\cite{wang2025timestamp,gollapudi2023filtered}. In this design, temporal information is used mainly to restrict the search space, while the underlying vector search objective remains unchanged and continues to rank candidates purely by semantic similarity. The second is \emph{rerank-based} search, which performs semantic search first and then incorporates temporal signals through reranking or score fusion~\cite{dong2010towards,openclaw_memory_search}. Such methods retain a standard semantic retriever as the backbone and introduce time only after an initial candidate set has been produced. 
Overall, current vector search solutions still treat time as an auxiliary factor layered on top of semantic search, rather than as a native component of the search objective.

\textbf{Limitations and Motivation.}
Despite these efforts, existing vector search solutions remain inadequate for dynamic workloads in which relevance derived from semantic similarity changes continuously over time. 
Filter-based methods are too coarse to capture soft temporal decay, since relevance often decreases gradually rather than disappearing at a fixed timestamp threshold. 
Rerank-based pipelines are limited by the recall of the first-stage retriever: if an item with the right semantic--temporal trade-off is not retrieved initially, no downstream reranker can recover it. 

This gap calls for a new formulation and a corresponding system solution. Rather than adding temporal signals to an existing vector search pipeline, we ask a more fundamental and challenging question: \emph{how should vector search be reformulated when relevance changes continuously over time?} 
Addressing this question requires incorporating temporal decay directly into the search objective, so that indexing and search can be built around relevance jointly determined by semantics and time.
This motivates a new problem that we call \emph{time-decayed vector search}: given a query vector and a collection of timestamped data vectors, retrieve the items that are most relevant under a scoring function that combines semantic similarity with time decay.

\textbf{Our Solution.}
In this paper, we formalize \emph{Time-Decayed Vector Search} (TDVS) as a new vector search problem and present \emph{Chronos}, a TDVS-native framework, together with its hierarchical graph index \emph{TANGO} to efficiently support TDVS.
Specifically, we first formulate TDVS over timestamped vectors by incorporating continuous temporal decay directly into the search objective, so that relevance is jointly determined by semantic similarity and temporal freshness at query time (Section~\ref{sec:tdvs}).
We then design two solutions for TDVS: a generic solution \emph{Score-Preserving Temporal Reduction} (STR) (Section~\ref{sec:str}) and a TDVS-native solution \emph{Chronos} (Section~\ref{sec:chronos_framework}).
STR exactly reduces TDVS to conventional Maximum Inner Product Search (MIPS), enabling existing MIPS indexes to be directly reused to support TDVS, but it inherits structural constraints from the temporal factorization.

To overcome these constraints, we present the \emph{Chronos} framework, which derives an exact metric formulation for unit-norm embeddings and introduces \emph{Query-Orthogonal TimeLift} to control data--data geometry, while preserving all query--data scores and rankings.
Building on Chronos, we propose \emph{TANGO} (Section~\ref{sec:tango}), short for \emph{\underline{T}ime-\underline{A}ware \underline{N}avigable \underline{G}raph with Query-\underline{O}rthogonal TimeLift}. 
TANGO is a TDVS-native hierarchical graph index that assigns different TimeLift geometries to different graph layers, enabling upper layers to strengthen long-range semantic connectivity while the base layer preserves temporal locality. TANGO traverses the hierarchy under the exact TDVS score, caches temporal factors to eliminate repeated exponentiation, and supports efficient online insertion. 
Extensive experiments on seven real-world datasets demonstrate that TANGO consistently outperforms state-of-the-art graph-based competitors in both query and index performance while remaining robust across diverse temporal settings (Section~\ref{sec:experiments}).

Our contributions are summarized as follows: 

\noindent\cbullet We formalize \emph{time-decayed vector search} (TDVS) and study two representative objectives that combine semantic similarity and temporal freshness additively or multiplicatively.
We further design \emph{Score-Preserving Temporal Reduction} (STR), an exact score-preserving reduction that converts TDVS into Maximum Inner Product Search (MIPS), enabling existing MIPS indexes to directly support TDVS.

\noindent\cbullet We present \emph{Chronos}, a TDVS-native framework that provides an exact metricization of the complete TDVS objective.
Its \emph{Query-Orthogonal TimeLift} establishes a family of query-equivalent, but geometry-distinct metric spaces, enabling controllable semantic--temporal geometry, while exactly preserving every TDVS score and ranking.

\noindent\cbullet We propose \emph{TANGO}, a TDVS-native hierarchical graph index built on Chronos. TANGO realizes layer-specific semantic--temporal geometries within a unified graph, jointly supporting temporal
locality and long-range semantic connectivity, and at the same time, enabling efficient online insertion.

\noindent\cbullet We conduct extensive experiments comparing TANGO with state-of-the-art graph-based methods.
TANGO achieves up to $3.5\times$ higher query throughput and $4.05\times$ faster index construction than the best-performing competitors.
Across diverse temporal settings, TANGO still outperforms all competitors and supports efficient online insertion, demonstrating its robustness and practicality.

\begin{table}[t]
  \centering
  \caption{Summary of key notations.}
  \vspace*{-0.2cm}
  \label{tab:notations}
  \renewcommand{\arraystretch}{1.02}
  \setlength{\tabcolsep}{3pt}
  \resizebox{\columnwidth}{!}{%
  \begin{tabular}{@{}ll@{}}
    \toprule
    \textbf{Notation} & \textbf{Description} \\
    \midrule
    $\mathcal D,n,d$
    & Dataset, cardinality, and vector dimension \\

    $(x_i,t_i),(q,\tau)$
    & Data and query vector--time pairs \\

    $s(\cdot,\cdot),\Delta_i,w(\cdot)$
    & Semantic similarity, item age, and time-decay function \\

    $\lambda,h$
    & Decay rate and half-life, with $\lambda=\ln 2/h$ \\

    $F_m,\alpha,k$
    & TDVS score, additive semantic weight, and result size \\

    $\psi_{Q/X},T,z_i^{(T)},g_\tau^{(T)}$
    & STR encoders, anchor, and temporal factors \\

    $\widetilde q_\tau^m,\widetilde x_i^m$
    & STR mapped query/data vectors \\

    $k_\lambda,\phi_\lambda,Q_m,P_m,\mathcal H_m$
    & Chronos kernel, feature map, mappings, and Hilbert space \\

    $c_{ij},k_{ij}$
    & Pairwise semantic and temporal affinities \\
    
    $d_m,D_m$
    & Hilbert-space metric and its squared form \\  
    $\kappa,\kappa_\ell,d_m^{(\kappa)},D_m^{(\kappa)}$
    & TimeLift parameter, layer value, metric, and squared form \\

    $\mathcal G_\ell$
    & TANGO graph at layer $\ell$ \\

    $t_0,a_\tau,b_i$
    & Reference time and cached query/data factors \\
    \bottomrule
  \end{tabular}%
  }
  \vspace*{-0.3cm}
\end{table}

\begin{figure*}[t]
    \vspace*{-1.1cm}
    \centering

    \begin{minipage}{0.95\textwidth}
        \centering

        \subfigure[Candidate points with half-life $h=30$.]{
            \includegraphics[
                width=0.255\linewidth
            ]{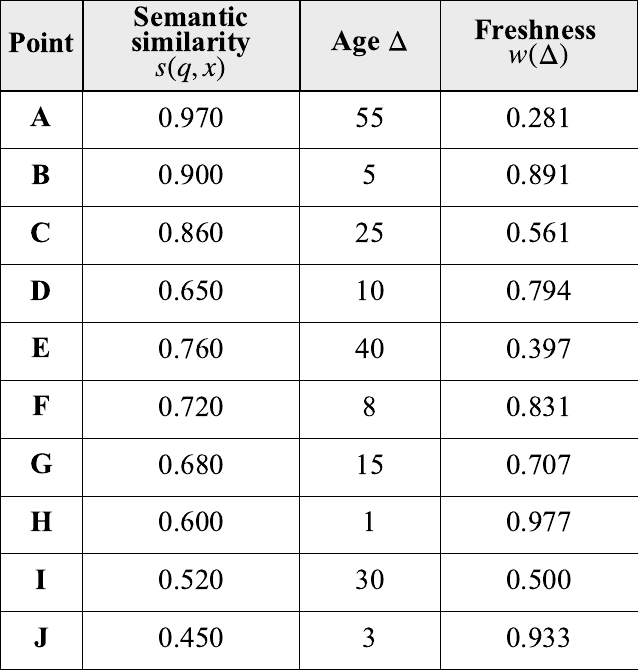}
            \label{fig:tdvs-example-statistics}
        }
        \hfill
        \subfigure[Distribution in semantic similarity--freshness space. 
        ]{
            \includegraphics[
                width=0.32\linewidth
            ]{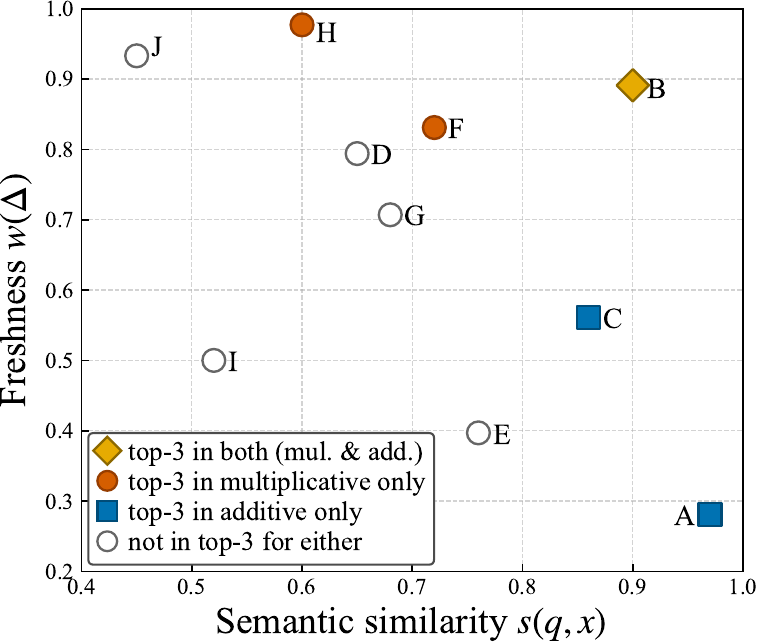}
            \label{fig:tdvs-example-scatter}
        }
        \hfill
        \subfigure[Scores and rankings under the two TDVS modes.]{
            \includegraphics[
                width=0.30\linewidth
            ]{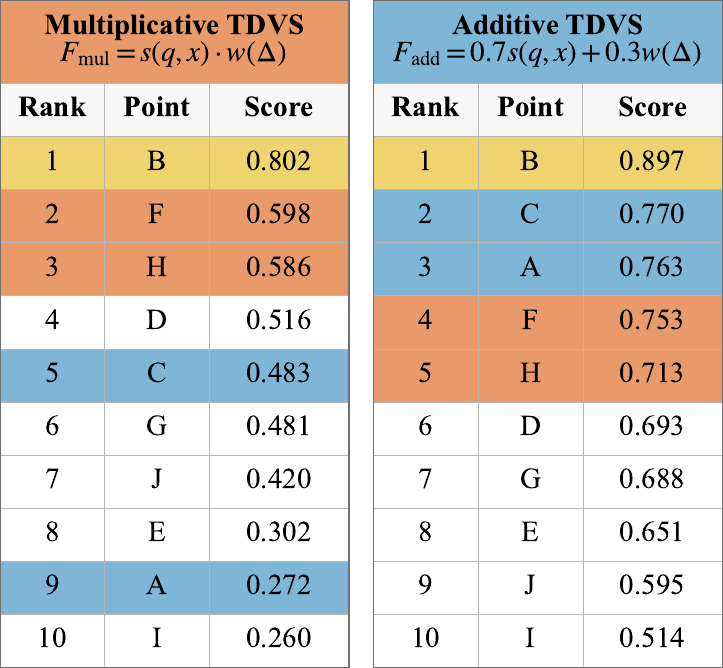}
            \label{fig:tdvs-example-rankings}
        }

    \end{minipage}
    \vspace*{-0.1cm}
    \caption{An illustrative example of time-decayed vector search: (a) candidates' semantic similarity to query and temporal freshness (with half-life $h=30$); (b) and (c) TDVS ranking results: (c) reports the scores and rankings under multiplicative TDVS, and additive ($\alpha=0.7$) TDVS, while (b) visualizes these results in the semantic similarity--freshness space, where yellow, orange, blue, and white denote top-3 membership under both modes, multiplicative only, additive only, and neither, respectively.}
    \label{fig:tdvs-example}
    \vspace*{-0.2cm}
\end{figure*}

\section{Problem Statement}

\subsection{Classical Vector Search}

We begin with the standard vector search problem. Given a dataset $\mathcal{D}=\{x_i\}_{i=1}^n$, where $x_i\in\mathbb{R}^d$, a query $q\in\mathbb{R}^d$, a similarity function $s(\cdot,\cdot)$, and an
integer $k$, \emph{exact vector search} returns the $k$ data vectors with the highest similarity to the query:
$
\mathcal N_k(q)
=
\operatorname{TopK}_{x_i\in\mathcal D}^{k}
s(q,x_i).
$
Table~\ref{tab:notations} summarizes the key notation used throughout the paper.

In high-dimensional settings, exact vector search becomes challenging due to the curse of dimensionality~\cite{hinneburg2000nearest,borodin1999lower}; \emph{approximate vector search} is therefore the more practical choice, returning an approximate top-$k$ set $\widehat{\mathcal{N}}_k(q)$ at substantially lower retrieval cost~\cite{hydra1, hydra2}, with answer quality measured by
$\mathrm{Recall}@k = \frac{|\widehat{\mathcal{N}}_k(q)\cap \mathcal{N}_k(q)|}{k}$.

% We begin with the standard vector search problem. Let $\mathcal{D}=\{x_i\}_{i=1}^n$ be a collection of data vectors, where $x_i\in\mathbb{R}^d$, and let $q\in\mathbb{R}^d$ be a query vector. Given a similarity function $s(q,x_i)$, the goal of vector search is to retrieve the items that are most similar to the query.

% \begin{definition}[Exact Vector Search]
% Given a dataset $\mathcal{D}=\{x_i\}_{i=1}^n$, a query vector $q$, a similarity function $s(\cdot,\cdot)$, and an integer $k$, the exact vector search problem returns the top-$k$ set
% $
% \mathcal N_k(q)
% =
% \operatorname{TopK}_{x_i\in\mathcal D}^{k}
% s(q,x_i).
% $
% \end{definition}

% In high-dimensional settings, exact vector search becomes challenging because of the \emph{curse of dimensionality} phenomenon~\cite{hinneburg2000nearest,borodin1999lower}, making approximate vector search the more practical choice, whihch returns an approximate top-$k$ set $\widehat{\mathcal{N}}_k(q)$ at substantially lower retrieval cost, with answer quality commonly measured by
% $\mathrm{Recall}@k=|\widehat{\mathcal{N}}_k(q)\cap\mathcal{N}_k(q)|/k$.

% \begin{definition}[Approximate Vector Search]
% Given the same input as exact vector search, the approximate vector search problem returns an approximate top-$k$ set $\widehat{\mathcal{N}}_k(q)$ that aims to remain close to $\mathcal{N}_k(q)$ while substantially reducing retrieval cost.
% \end{definition}

% A common evaluation metric is $\mathrm{Recall}@k = \frac{|\widehat{\mathcal{N}}_k(q)\cap \mathcal{N}_k(q)|}{k}$.

Classical vector search, whether exact or approximate, assumes that relevance is fully determined by semantic similarity and does not change once an item's embedding is fixed. This assumption is appropriate for many workloads, but it becomes insufficient when the usefulness of an item changes with time~\cite{park2023generative,qian2024timer4,kasai2023realtime,ryu2025news}.

\subsection{Time-Decayed Vector Search} \label{sec:tdvs}

We now introduce \emph{time-decayed vector search}, where each data vector is associated with a timestamp. Formally, let $\mathcal{D}=\{(x_i,t_i)\}_{i=1}^n$, where $x_i\in\mathbb{R}^d$ is the vector of item $i$, and $t_i$ denotes the time at which the item is created. For a query issued at time $\tau$, we define the age of item $i$ as $\Delta_i(\tau)=\tau-t_i$, where $\Delta_i(\tau)\ge 0$.

\begin{definition}[Time-Decay Function]
A time-decay function is a mapping
$w:\mathbb{R}_{\geq 0}\rightarrow(0,1]$ that assigns a freshness
weight to a data vector according to its age $\Delta$, with
$w(0)=1$ and $w(\Delta)$ monotonically non-increasing in $\Delta$.
\end{definition}

We use the exponential decay function
\[
w_\lambda(\Delta)=e^{-\lambda\Delta},
\qquad \lambda\geq0,
\]
where $\lambda$ controls the decay rate.
The exponential family is well motivated in both theory and practice. It is commonly used to model the \emph{Ebbinghaus Forgetting Curve}~\cite{murre2015replication} and admits the interpretable half-life parameterization
\[
w_h(\Delta)=e^{-(\ln 2/h)\Delta}=2^{-\Delta/h},
\qquad h>0,
\]
which is equivalent to $w_\lambda(\Delta)$ when $\lambda=\ln 2/h$~\cite{settles2016trainable}. Similar decay functions are used for recency modeling in Generative Agents and OpenClaw~\cite{park2023generative,openclaw_memory_search}, although typically as post-retrieval scoring, or reranking components, rather than as a native search objective. 

% This form is supported by both theory and practice. From the theoretical side, such behavior is commonly described in memory science by \emph{Ebbinghaus Forgetting Curve}, which characterizes how retention decreases as elapsed time grows~\cite{murre2015replication}. 
% Exponential decay is one of the standard forms used to model this process, and half-life models can be viewed as an equivalent reparameterization in which the decay rate is expressed by the time required for the utility to drop by half~\cite{settles2016trainable}. Specifically, the half-life form
% \[
% w_h(\Delta)=e^{-(\ln 2 / h)\Delta}=2^{-\Delta/h},
% \]
% where $h>0$ denotes the half-life. This parameterization is equivalent to $w_\lambda(\Delta)=e^{-\lambda\Delta}$ when $\lambda=\ln 2/h$.

% From the practical side, several systems already model recency using this family of functions. For example, Generative Agents uses an exponentially decayed recency component in memory retrieval~\cite{park2023generative}, while OpenClaw exposes a half-life-based temporal decay option in memory search~\cite{openclaw_memory_search}. Beyond retrieval, half-life models have also been adopted in practical learning systems to characterize time-dependent recall and utility~\cite{settles2016trainable}. 
% These examples suggest that the exponential family and its half-life parameterization are already compatible with practical system design, although current systems typically use them as post-retrieval scoring adjustments or reranking components rather than as a native search objective.

\begin{definition}[Time-Decayed Vector Search]
Given a timestamped vector dataset $\mathcal{D}=\{(x_i,t_i)\}_{i=1}^n$, a query vector $q$, a query time $\tau$, a similarity function $s(\cdot,\cdot)$, and a time-decay function $w(\cdot)$, the time-decayed vector search (TDVS) problem retrieves the top-$k$ items under a joint semantic-temporal score
$
F\!\left(s(q,x_i),\,w(\Delta_i(\tau))\right),
$ 
where $\Delta_i(\tau)=\tau-t_i$ denotes the age of item $i$. Formally,
\[
\mathcal N_k^{\mathrm{tdvs}}(q,\tau)
=
\operatorname{TopK}_{(x_i,t_i)\in\mathcal D}^{k}
F\!\left(
s(q,x_i),
w(\Delta_i(\tau))
\right).
\]
\end{definition}

% Like classical vector search, time-decayed vector search (TDVS) can be studied in both exact and approximate forms. Exact TDVS returns the true top-$k$ set $\mathcal{N}^{\mathrm{tdvs}}_k(q,\tau)$ under $F$, while approximate TDVS aims to retrieve a high-quality approximation $\widehat{\mathcal{N}}^{\mathrm{tdvs}}_k(q,\tau)$ at substantially lower search cost.

In this paper, we study two fundamental and practically important formulations of TDVS, \emph{additive} and \emph{multiplicative} time-decayed vector search. 
These two formulations capture distinct ways in which temporal freshness interacts with semantic similarity.

\begin{definition}[Additive Time-Decayed Vector Search]
\label{def:additive-tdvs}
The \emph{additive} form of TDVS is defined by the scoring function
\[
F_{\mathrm{add}}(q,x_i,\tau)
=
\alpha\,s(q,x_i) 
+
(1-\alpha)\,w(\Delta_i(\tau)),
\qquad \alpha \in \left[0, 1\right].
\] 
\vspace*{-0.3cm}
\end{definition}

% The additive form treats temporal freshness as an additional preference signal that complements semantic similarity. The semantic term $s(q,x_i)$ measures content match, while the temporal term $w(\Delta_i(\tau))$ gives an extra advantage to newer items. The coefficient $\alpha$ controls the relative importance of these two factors. Therefore, this form is suitable for settings in which time matters but should not dominate semantic relevance. In such cases, older items may still remain competitive if their semantic similarity is sufficiently strong, while newer items benefit from recency without fully overriding semantic evidence. For example, in agent systems, this form is more suitable for managing memory types whose usefulness is not tightly coupled to the current state, such as long-term contextual memory, where recency is treated as an additional preference rather than as a strict decay factor~\cite{park2023generative}. More broadly, the additive form is well suited to soft-preference settings, where time acts as a bias toward freshness.

\begin{definition}[Multiplicative Time-Decayed Vector Search]
The \emph{multiplicative} form of TDVS is defined by the scoring function\footnote{When the semantic score is signed, multiplicative TDVS attenuates its magnitude toward the neutral value zero while preserving its sign. This behavior is natural because positive and negative scores represent opposite but meaningful semantic associations: as an item ages, the strength of either association should diminish, while its direction should remain unchanged.
Moreover, whenever a query has at least $k$ positive-score items, negative-score items cannot appear in its exact TDVS top-$k$ result, because the strictly positive decay factor preserves score signs.}
\label{def:multiplicative-tdvs}
\[
F_{\mathrm{mul}}(q,x_i,\tau)
=
s(q,x_i) \, w(\Delta_i(\tau)).
\]
\vspace*{-0.4cm}
\end{definition}

Note that in practice, the TDVS parameters are straightforward to configure: the half-life $h$ directly specifies the application-specific timescale at which freshness is reduced by half, while the additive weight $\alpha$ controls the desired semantic--temporal trade-off. Similar controls are already used in existing systems:
OpenClaw adopts a 30-day recency half-life~\cite{openclaw_memory_search}, while Qdrant exposes configurable decay timescales and score weights~\cite{qdrant_search_relevance}. When relevance feedback is available, $h$ and $\alpha$ can also be tuned on held-out queries.

% The multiplicative form imposes a tighter coupling between semantic similarity and temporal freshness. Instead of adding freshness as a separate preference term, it directly scales the semantic contribution by the decay factor. As a result, even a highly similar item can become substantially less competitive if it is sufficiently old. Therefore, this form is more appropriate for settings in which freshness is a stronger requirement and stale information should be discounted more aggressively. Representative examples include freshness-sensitive RAG and temporal question answering, where older passages may still be semantically related to the query but should often be ranked below more recent evidence as the context evolves~\cite{qian2024timer4,kasai2023realtime}. Within agent systems, this form is more suitable for managing memory types whose usefulness is tightly coupled to the current state, such as episodic memory~\cite{openclaw_memory_search}. Compared with the additive form, the multiplicative form yields a sharper notion of temporal relevance, since the magnitude of the semantic contribution decreases
% continuously with age.

Figure~\ref{fig:tdvs-example} illustrates how semantic similarity and temporal freshness jointly affect TDVS ranking (in this example, half-life of $h=30$ and additive weight $\alpha=0.7$).
Point $B$, which has both high semantic similarity and high
freshness, ranks first under both objectives.
Multiplicative TDVS promotes the fresher points $F$ and $H$ into the top-$3$, while substantially lowering the rank of the highly similar, but old, point $A$.
In contrast, additive TDVS retains $C$ and $A$ in the top-$3$ because semantic similarity receives the larger weight.

Figure~\ref{fig:tdvs-example-rankings} demonstrates that the two TDVS modes are suitable for different freshness requirements.
Additive TDVS treats temporal freshness as a complementary preference signal, with $\alpha$ controlling the semantic--temporal trade-off.
It is therefore appropriate when time matters, but strong semantic matches should remain competitive, such as long-term contextual memories, whose usefulness is not tightly coupled to an agent's current state~\cite{park2023generative}.
Multiplicative TDVS instead scales semantic relevance directly by freshness and discounts stale items more aggressively, making it more suitable for freshness-sensitive RAG, temporal question answering, and state-dependent episodic memories~\cite{qian2024timer4,kasai2023realtime,
openclaw_memory_search}.
Thus, the additive form models a soft preference for recency, whereas the multiplicative form imposes a tighter coupling between semantic relevance and temporal freshness.

\subsection{STR: Exact TDVS-to-MIPS Reduction}
\label{sec:str}

The TDVS formulation in Section~\ref{sec:tdvs} is defined using a general semantic similarity function $s(\cdot,\cdot)$. 
Here, a \emph{semantic measure} refers to the underlying query--data measure used in vector search to quantify semantic similarity; common choices include inner product, cosine similarity, and Euclidean distance.
Rather than developing a separate TDVS solution for each semantic measure, we investigate whether their additive and multiplicative TDVS objectives can be reduced to a common computational primitive.

Through mathematical derivation, we show that for each of these three semantic measures, the additive and multiplicative TDVS score functions defined in Definitions~\ref{def:additive-tdvs} and~\ref{def:multiplicative-tdvs} can be represented exactly as finite-dimensional inner products.
Consequently, TDVS can be reduced to the Maximum Inner Product Search (MIPS) problem while preserving all query--data scores.
We call this exact TDVS-to-MIPS reduction
\emph{Score-Preserving Temporal Reduction} (\emph{STR}). 

STR proceeds in three steps: semantic encoding, temporal factorization, and vector mapping. Semantic encoding converts query and data vectors so that their inner product reproduces the original semantic score. Temporal factorization decomposes exponential decay into query- and data-side factors. Vector mapping then incorporates these temporal factors into the encoded semantic features, producing mode-specific query and data vectors whose inner product equals the complete TDVS score.

\paragraph{Step 1: Semantic encoding.}
For each semantic measure, we define a query encoder $\psi_Q$ and a
data encoder $\psi_X$. Given a query $q$ and a data vector $x_i$,
these encoders produce semantic feature vectors $\psi_Q(q)$ and
$\psi_X(x_i)$ such that
\begin{equation}
    s(q,x_i)
    =
    \left\langle
        \psi_Q(q),
        \psi_X(x_i)
    \right\rangle.
    \label{eq:semantic-ip-encoding}
\end{equation}

For inner-product semantics, no additional encoding is required: $\psi_Q(q)=q$ and $\psi_X(x_i)=x_i$.
For cosine similarity over nonzero vectors, the encoded features are their L2-normalized representations:
$\psi_Q(q)=\frac{q}{\|q\|_2},\psi_X(x_i)=\frac{x_i}{\|x_i\|_2}.$
% The data vectors can be normalized once during index construction, while each query is normalized once before search.
For Euclidean semantics, let $\rho$ be a fixed normalization scale such that $\|q-x_i\|_2\leq \rho$ for every valid query--data pair.
We use the following ranking-equivalent transformation of squared Euclidean distance:
$s_{\mathrm{L2}}(q,x_i)=1-\frac{2}{\rho^2}\|q-x_i\|_2^2$, which lies in $[-1,1]$ and maintains the same ranking as the Euclidean distance.
Its semantic encoding is
\begin{equation}
    \psi_Q(q)
    =
    \begin{bmatrix}
        \frac{2}{\rho}q\\
        1-\frac{2}{\rho^2}\|q\|_2^2\\
        1
    \end{bmatrix},
    \qquad
    \psi_X(x_i)
    =
    \begin{bmatrix}
        \frac{2}{\rho}x_i\\
        1\\
        -\frac{2}{\rho^2}\|x_i\|_2^2
    \end{bmatrix}.
    \label{eq:l2-semantic-encoding}
\end{equation}
% The squared norm of each data vector can be precomputed once, while the query norm is computed once per query.

\paragraph{Step 2: Temporal factorization.}
Let $T$ be a temporal anchor; define a data-side temporal factor $z_i^{(T)}$, and a query-side temporal factor $g_\tau^{(T)}$ as
\begin{equation}
    z_i^{(T)}
    =
    e^{-\lambda(T-t_i)},
    \qquad
    g_\tau^{(T)}
    =
    e^{-\lambda(\tau-T)}.
    \label{eq:str-factors}
\end{equation}
Their product recovers the original decay exactly:
\begin{equation}
    g_\tau^{(T)}z_i^{(T)}
    =
    e^{-\lambda(\tau-t_i)}.
    \label{eq:str-factorization}
\end{equation}
Although $T$ can be chosen arbitrarily, a common choice at index
construction time is $T=\max_i t_i$, which ensures
$0<z_i^{(T)}\leq1$ for every data vector in the initial dataset.

\paragraph{Step 3: Vector mapping.}
We now combine the encoded semantic features from Step~1 with the temporal factors from Step~2. For each TDVS mode $m\in\{\mathrm{mul},\mathrm{add}\}$, STR defines a query mapping $\Phi_{Q,m}$ and a data mapping $\Phi_{X,m}$. Applying these mappings to a query $(q,\tau)$ and a timestamped data vector $(x_i,t_i)$ produces the corresponding STR query representation $\widetilde q_{\tau}^{m}=\Phi_{Q,m}(q,\tau)$ and data representation $\widetilde x_i^{m}=\Phi_{X,m}(x_i,t_i)$.

Since the differences among inner product, cosine similarity, and ranking-equivalent transformation of squared Euclidean distance have already been captured by the semantic encoders $\psi_Q$ and $\psi_X$, the query and data mappings apply uniformly to all three semantic measures.
For multiplicative TDVS, the mapping is defined as:
\begin{equation}
    \Phi_{Q,\mathrm{mul}}(q,\tau)
    =
    g_\tau^{(T)}\psi_Q(q),
    \qquad
    \Phi_{X,\mathrm{mul}}(x_i,t_i)
    =
    z_i^{(T)}\psi_X(x_i).
    \label{eq:str-mul-mapping}
\end{equation}
For additive TDVS, the mapping is defined as:
\begin{equation}
    \Phi_{Q,\mathrm{add}}(q,\tau)
    =
    \begin{bmatrix}
        \sqrt{\alpha}\,\psi_Q(q)\\
        \sqrt{1-\alpha}\,g_\tau^{(T)}
    \end{bmatrix},
    \,\,\,\,
    \Phi_{X,\mathrm{add}}(x_i,t_i)
    =
    \begin{bmatrix}
        \sqrt{\alpha}\,\psi_X(x_i)\\
        \sqrt{1-\alpha}\,z_i^{(T)}
    \end{bmatrix}.
    \label{eq:str-add-mapping}
\end{equation}
By Equations~\eqref{eq:semantic-ip-encoding} and
\eqref{eq:str-factorization}, the resulting representations satisfy
\begin{equation}
\begin{aligned}
    \left\langle
        \widetilde q_{\tau}^{\mathrm{mul}},
        \widetilde x_i^{\mathrm{mul}}
    \right\rangle=
        F_{\mathrm{mul}}(q,x_i,\tau),
    \,\,\,\,\,\,
    \left\langle
        \widetilde q_{\tau}^{\mathrm{add}},
        \widetilde x_i^{\mathrm{add}}
    \right\rangle=
        F_{\mathrm{add}}(q,x_i,\tau).
\end{aligned}
\label{eq:str-exactness}
\end{equation}
Equation~\eqref{eq:str-exactness} shows that the inner product between the mapped query and data representations is numerically identical to the target TDVS score. Consequently, STR introduces no score approximation and preserves the exact TDVS ranking.

STR provides a generic and practical solution to TDVS: with only lightweight query and data transformations, it enables existing MIPS indexes to directly support TDVS without changing their search objectives.
We instantiate STR with representative graph-based MIPS indexes~\cite{morozov2018non,liu2020understanding,tan2021norm,chen2025maximum,chen2025stitching}, which serve as competitive baselines in our experimental evaluation (Section~\ref{sec:experiments}).

Despite its exactness and practical convenience, STR inherits structural constraints from temporal factorization. At the representation level, transformed vectors depend on the global anchor $T$. Keeping $T$ fixed causes newly inserted vectors with $t_i>T$ to satisfy $z_i^{(T)}>1$, progressively widening the range of transformed norms; updating $T$, however, changes existing representations and may require re-encoding the stored vectors and rebuilding the index. At the geometry level, the pairwise temporal affinity
$z_i^{(T)}z_j^{(T)}=e^{-\lambda(2T-(t_i+t_j))}$
depends on the vectors' absolute freshness relative to $T$. This can make recent vectors more likely to become central hubs in the graph, which can distort semantic neighborhoods and make graph navigation more challenging.

\section{Chronos Framework} \label{sec:chronos_framework}

%\paragraph{From STR to Chronos.}
To overcome the structural constraints of STR, we propose \emph{Chronos}, a TDVS-native framework that removes the temporal anchor and decouples exact query--data score preservation from data--data geometry, enabling a better balance between temporal freshness and semantic similarity during indexing.
Figure~\ref{fig:tdvs-solution-overview} contrasts the two solution pipelines for TDVS. STR exactly reduces TDVS to MIPS and reuses state-of-the-art MIPS indexes, whereas Chronos provides a TDVS-native framework instantiated by TANGO. Both pipelines preserve the exact TDVS objective, but differ in whether the indexing geometry is inherited from generic MIPS methods or designed specifically for TDVS.

\begin{figure}[tb]
    \centering
    \includegraphics[width=\linewidth]
    {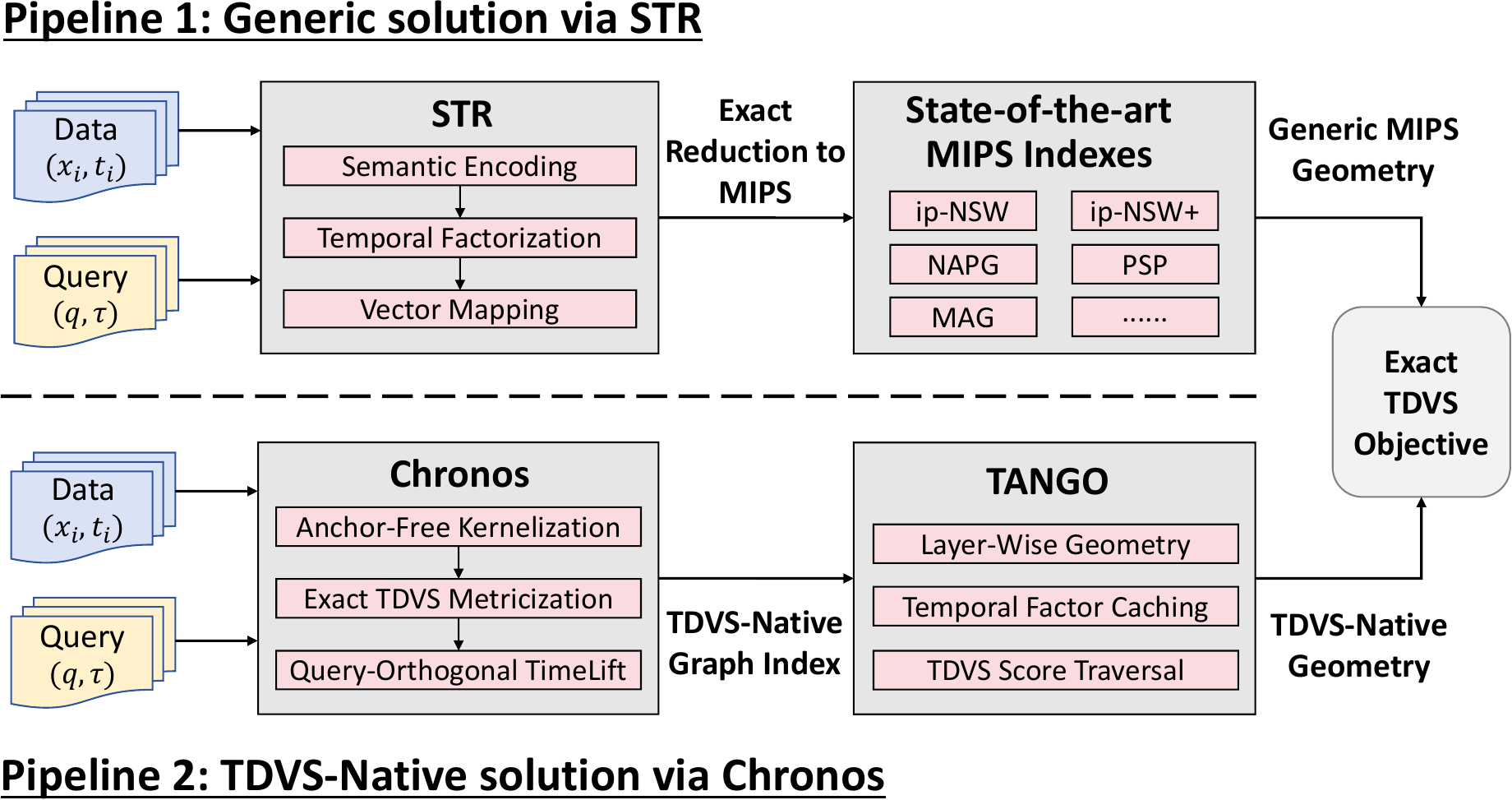}
    \caption{Overview of STR and Chronos for TDVS.}
    \label{fig:tdvs-solution-overview}
    \vspace*{-0.3cm}
\end{figure}

\paragraph{Chronos Overview.}
Chronos is a TDVS-native framework that preserves exact query--data scores, while providing controllable data--data geometry for efficient indexing. It has three components: an anchor-free temporal representation that recovers exponential decay from relative timestamps (Section~\ref{sec:anchor_free_temporal}); an exact metricization that maps queries and data to constant-norm, score-preserving representations, converting TDVS ranking into nearest-neighbor search (Section~\ref{sec:chronos_metricization}); and a controllable geometry adjustment that balances temporal freshness and semantic similarity without altering query--data scores (Section~\ref{sec:timelift}). Chronos realizes these components efficiently by evaluating query--data and data--data distances directly from the original embeddings and timestamps, without materializing the mapped representations (Section~\ref{sec:chronos_realization}).

\paragraph{Embedding regime.}
Modern embedding models and retrieval pipelines increasingly use normalized representations, projecting embeddings onto the unit sphere at model output, or before indexing~\cite{openai_embeddings_documentation, reimers2019sentence, zhang2025qwen3}. 
Normalization removes vector magnitude as an additional signal and unifies common similarity measures. For unit-norm embeddings,
$
\langle q,x_i\rangle
=
\cos(q,x_i)
=
1-\frac{1}{2}\|q-x_i\|_2^2.
$
Thus, inner product, cosine similarity, and canonical unit-sphere Euclidean similarity\footnote{Euclidean retrieval minimizes $\|q-x_i\|_2$. STR uses the
ranking-equivalent similarity $1-2\|q-x_i\|_2^2/\rho^2$. Under Chronos's unit-norm embedding regime, $\rho=2$, and this score reduces to $1-\|q-x_i\|_2^2/2=\langle q,x_i\rangle$.} induce identical TDVS objectives. Chronos targets this important regime and provides a unified metricization framework. 
Chronos can be extended to non-normalized embeddings through norm completion, which introduces an additional norm-dependent geometric component orthogonal to the temporal design studied here. We leave this extension to future work.

\subsection{Anchor-Free Temporal Representation}
\label{sec:anchor_free_temporal}
TDVS directly specifies only the temporal decay for a valid query--data pair, namely $e^{-\lambda(\tau-t_i)}$ for $\tau\geq t_i$.
However, index construction must also compare two data vectors and therefore requires a symmetric pairwise temporal affinity between their timestamps $t_i$ and $t_j$. 
STR implicitly induces an anchor-dependent pairwise temporal affinity $e^{-\lambda(2T-(t_i+t_j))}$, which measures the joint freshness of the two data vectors relative to $T$.
Chronos introduces an anchor-free form $e^{-\lambda|t_i-t_j|}$, which depends only on their relative temporal separation.
Figure~\ref{fig:anchor-free-temporal} illustrates this distinction. STR relates both timestamps to the global anchor $T$, whereas Chronos directly relates them through $|t_i-t_j|$, eliminating the anchor from the pairwise temporal geometry.
To obtain such a temporal affinity while preserving the original query--data decay, Chronos presents the \emph{Anchor-Free Temporal Kernel}.

\begin{figure}[tb]
    \centering
    \includegraphics[width=0.80\linewidth]
    {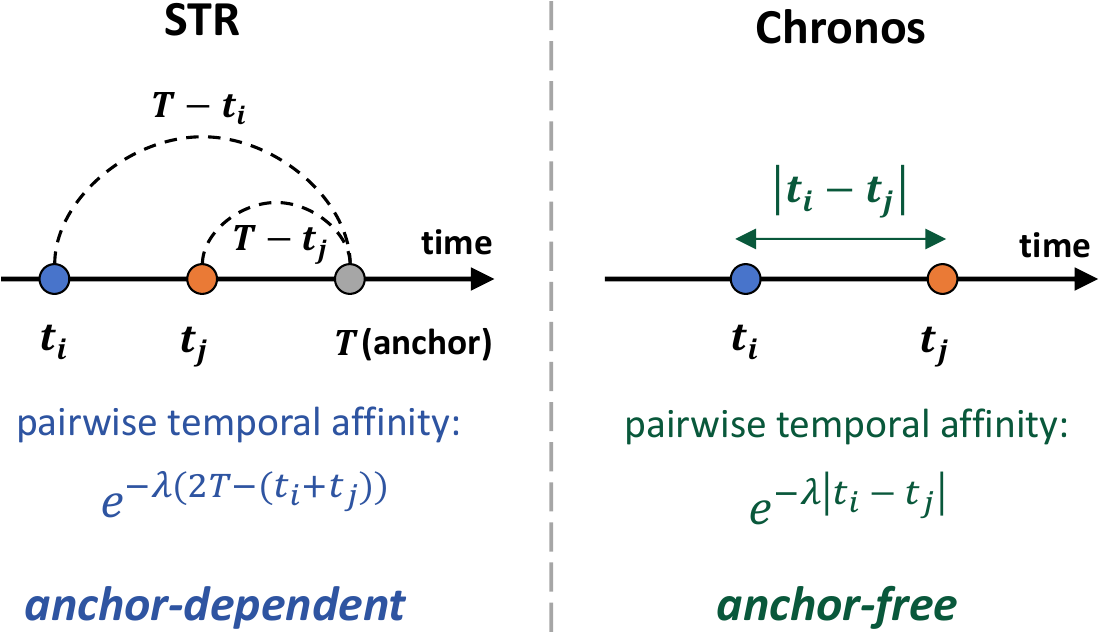}
    \caption{Anchor-dependent and anchor-free pairwise temporal affinities.}
    \label{fig:anchor-free-temporal}
\end{figure}

\begin{definition}[Anchor-Free Temporal Kernel]
\label{def:anchor-free-temporal-kernel}
For timestamps $s,t\in\mathbb{R}$, Chronos defines the anchor-free
temporal kernel as
\begin{equation}
    k_\lambda(s,t)
    =
    e^{-\lambda|s-t|},
    \qquad \lambda\geq0.
    \label{eq:chronos-temporal-kernel}
\end{equation}
\end{definition}

This kernel is the standard one-dimensional Laplacian kernel applied to timestamps.
For every valid query--data pair $\tau\geq t_i$, the kernel exactly recovers the original temporal decay: $k_\lambda(\tau,t_i)=e^{-\lambda|\tau-t_i|}=e^{-\lambda(\tau-t_i)}.$
Moreover, the kernel is symmetric, $k_\lambda(s,t)=k_\lambda(t,s)$, and depends only on the relative temporal separation $|s-t|$.
It is also translation invariant:
$k_\lambda(s+c,t+c)=k_\lambda(s,t),c\in\mathbb{R}.$ Hence, shifting the time origin, or inserting data with later timestamps does not alter the temporal affinities among indexed vectors. 

To use this temporal kernel as a geometric building block for TDVS, we establish its geometric realization in a Hilbert space: Lemma~\ref{lem:temporal-kernel-pd} (below) shows that there exists a unit-norm feature map whose pairwise inner products exactly recover $k_\lambda$.
This representation allows the temporal component to be combined with unit-norm semantic embeddings while preserving the TDVS score and, as shown in Section~\ref{sec:chronos_metricization}, enables its exact metricization.

\begin{lemma}[Unit-Norm Hilbert embedding of the temporal kernel]
\label{lem:temporal-kernel-pd}
For every $\lambda\geq0$, the anchor-free temporal kernel
$k_\lambda$ is positive semidefinite. Consequently, there exists a Hilbert space $\mathcal{H}_\lambda$, together with a feature map $\phi_\lambda:\mathbb{R}\rightarrow\mathcal{H}_\lambda$, such that for all $s,t\in\mathbb{R}$,
\begin{equation}
    \left\langle
        \phi_\lambda(s),
        \phi_\lambda(t)
    \right\rangle
    =
    k_\lambda(s,t),
    \,\,\,\,\,\,
    \|\phi_\lambda(s)\|=\|\phi_\lambda(t)\|=1.
    \label{eq:temporal-feature-map}
\end{equation}
\end{lemma}

\begin{proof}
For $\lambda>0$, the one-dimensional Laplacian kernel has the spectral representation
$
    e^{-\lambda|s-t|}
    =
    \int_{-\infty}^{\infty}
        e^{\mathrm{i}\omega(s-t)}
        \frac{\lambda}
             {\pi(\lambda^2+\omega^2)}
    \,d\omega.
$
Its spectral density $\lambda/(\pi(\lambda^2+\omega^2))$ is nonnegative. By Bochner's theorem, $k_\lambda$ is positive semidefinite. When $\lambda=0$, the kernel reduces to the constant kernel $k_0(s,t)=1$, which is also positive semidefinite.
Positive semidefiniteness guarantees the existence of the Hilbert space $\mathcal{H}_\lambda$ and the associated feature map $\phi_\lambda$. Moreover, for every $t\in\mathbb{R}$,
$
    \|\phi_\lambda(t)\|^2
    =
    k_\lambda(t,t)
    =
    1,
$
hence $\|\phi_\lambda(t)\|=1$.
\end{proof}

The feature map $\phi_\lambda$ is used only for theoretical construction and need not be explicitly constructed or stored; all required inner products can be evaluated exactly using Equation~\eqref{eq:chronos-temporal-kernel}. 
Section~\ref{sec:chronos_metricization} builds on this representation to derive the exact TDVS metricization.
% Section~\ref{sec:chronos_metricization} combines this unit-norm temporal representation with unit-norm semantic embeddings to derive exact score-preserving metric formulations for additive and multiplicative TDVS.

\subsection{Exact TDVS Metricization}
\label{sec:chronos_metricization}

By Lemma~\ref{lem:temporal-kernel-pd}, each timestamp $t$ has a unit-norm temporal feature $\phi_\lambda(t)$, while the semantic embeddings used by Chronos are also unit-norm. In this section, for each TDVS mode, Chronos combines the unit-norm semantic embeddings with the unit-norm temporal features to map queries and data vectors into a mode-specific Hilbert space. The resulting mappings preserve the complete TDVS score and assign unit norm to every mapped query and data vector. 
These properties are formalized in Theorem~\ref{thm:chronos-metricization} (below), which shows that the mapped inner product exactly recovers the TDVS score and that score maximization is therefore equivalent to nearest-neighbor search under the induced Hilbert-space metric. We refer to this exact equivalence as \emph{TDVS metricization}.

For multiplicative TDVS, the target score is the product of semantic similarity and temporal freshness. Chronos combines each semantic embedding with its corresponding temporal feature using the tensor product:
\begin{equation}
    Q_{\mathrm{mul}}(q,\tau)
    =
    q\otimes\phi_\lambda(\tau),
    \qquad
    P_{\mathrm{mul}}(x_i,t_i)
    =
    x_i\otimes\phi_\lambda(t_i),
    \label{eq:chronos-mul-mapping}
\end{equation}
where both lie in the tensor-product space $\mathcal{H}_{\mathrm{mul}}=\mathbb{R}^d\otimes\mathcal{H}_\lambda$. 

For additive TDVS, the target score is the weighted sum of semantic similarity and temporal freshness. Chronos combines each semantic embedding with its corresponding temporal feature using a weighted direct sum:
\begin{equation}
    Q_{\mathrm{add}}(q,\tau)
    =
    \begin{bmatrix}
        \sqrt{\alpha}\,q\\
        \sqrt{1-\alpha}\,\phi_\lambda(\tau)
    \end{bmatrix},
    \,\,
    P_{\mathrm{add}}(x_i,t_i)
    =
    \begin{bmatrix}
        \sqrt{\alpha}\,x_i\\
        \sqrt{1-\alpha}\,\phi_\lambda(t_i)
    \end{bmatrix}.
    \label{eq:chronos-add-mapping}
\end{equation}
Both representations lie in the direct-sum space
$\mathcal{H}_{\mathrm{add}}=\mathbb{R}^d\oplus\mathcal{H}_\lambda$. The square-root scaling
ensures that their inner product reproduces the weights $\alpha$ and
$1-\alpha$ in the additive TDVS objective.

In the remainder of this section, all inner products, norms, and distances are taken in the mode-specific Hilbert space $\mathcal{H}_m$.
We use $d_m$ to denote the Hilbert-space metric and $D_m=d_m^2$ to denote its squared form.
Since squaring preserves nonnegative order, $d_m$ and $D_m$ induce the same nearest-neighbor ordering.

\begin{theorem}[Exact TDVS Metricization]
\label{thm:chronos-metricization}
For each mode $m\in\{\mathrm{mul},\mathrm{add}\}$ and every valid query--data pair $\tau\geq t_i$, the Chronos mappings satisfy
$
    \left\langle
        Q_m(q,\tau),
        P_m(x_i,t_i)
    \right\rangle
    =
    F_m(q,x_i,\tau)
$
and
$
    \|Q_m(q,\tau)\|
    =
    \|P_m(x_i,t_i)\|
    =
    1.
$
TDVS score maximization is equivalent to
minimizing either the Hilbert-space metric $d_m$ or its squared form $D_m$.
\end{theorem}

\begin{proof}
For multiplicative TDVS, applying the tensor-product identity and Equation~\eqref{eq:temporal-feature-map} gives
$
    \left\langle Q_{\mathrm{mul}},P_{\mathrm{mul}}\right\rangle
    =
    \langle q,x_i\rangle
    \left\langle
        \phi_\lambda(\tau),\phi_\lambda(t_i)
    \right\rangle
    =
    \langle q,x_i\rangle e^{-\lambda(\tau-t_i)}
    =
    F_{\mathrm{mul}}(q,x_i,\tau).
$
For additive TDVS,
$
    \left\langle Q_{\mathrm{add}},P_{\mathrm{add}}\right\rangle
    =
    \alpha\langle q,x_i\rangle
    +(1-\alpha)
    \left\langle
        \phi_\lambda(\tau),\phi_\lambda(t_i)
    \right\rangle
    =
    \alpha\langle q,x_i\rangle
    +(1-\alpha)e^{-\lambda(\tau-t_i)}
    =
    F_{\mathrm{add}}(q,x_i,\tau).
$
Since the semantic embeddings and temporal features are unit-norm, the tensor-product norm gives
$\|Q_{\mathrm{mul}}\|=\|P_{\mathrm{mul}}\|=1$, while the direct-sum construction gives
$
    \|Q_{\mathrm{add}}\|^2
    =
    \|P_{\mathrm{add}}\|^2
    =
    \alpha+(1-\alpha)
    =
    1.
$
Finally, expanding the squared distance for either mode gives
\begin{equation}
    D_m(q,x_i)
    =
    \left(d_m(q,x_i)\right)^2
    =
    \left\|
    Q_m-P_m
    \right\|^2
    =
    2-2F_m(q,x_i,\tau).
    \label{eq:chronos-query-data-distance}
\end{equation}
Thus, minimizing the Hilbert-space distance $d_m$ or its squared form $D_m$ is exactly equivalent to maximizing the TDVS score.
\end{proof}

\paragraph{Base data geometry.}
Beyond guaranteeing exact query--data score preservation, the metricization also determines the pairwise distances among indexed data vectors, thereby defining the base data geometry used for index construction.

For two data vectors $(x_i,t_i)$ and $(x_j,t_j)$, let
$c_{ij}=\langle x_i,x_j\rangle$ denote their semantic similarity and
let $k_{ij}=k_\lambda(t_i,t_j)=e^{-\lambda|t_i-t_j|}$ denote their temporal affinity.
For multiplicative TDVS, the squared distance between the mapped data representations in the corresponding Hilbert space is
\begin{align}
    D_{\mathrm{mul}}^{\mathrm{base}}(x_i,x_j)
    &=
    \left\|
        P_{\mathrm{mul}}(x_i,t_i)
        -
        P_{\mathrm{mul}}(x_j,t_j)
    \right\|^2
    \nonumber\\
    &=
    2(1-k_{ij})
    +
    2k_{ij}(1-c_{ij}).
    \label{eq:chronos-base-mul-dd}
\end{align}
For additive TDVS, the corresponding squared distance is
\begin{align}
    D_{\mathrm{add}}^{\mathrm{base}}(x_i,x_j)
    &=
    \left\|
        P_{\mathrm{add}}(x_i,t_i)
        -
        P_{\mathrm{add}}(x_j,t_j)
    \right\|^2
    \nonumber\\
    &=
    2(1-\alpha)(1-k_{ij})
    +
    2\alpha(1-c_{ij}).
    \label{eq:chronos-base-add-dd}
\end{align}
These expressions show how the base geometry combines semantic similarity and temporal affinity under the two TDVS modes. In the additive case, the two contributions are explicitly weighted by $\alpha$ and $1-\alpha$. In the multiplicative case, the semantic component $2k_{ij}(1-c_{ij})$ is coupled to the temporal affinity $k_{ij}$. As the temporal separation $|t_i-t_j|$ increases, $k_{ij}$ decreases, gradually reducing the influence of semantic similarity. When $k_{ij}$ approaches zero, the squared distance approaches $2$ regardless of $c_{ij}$, weakening the connections between semantically related vectors with distant timestamps. Thus, the base multiplicative geometry naturally favors temporal locality, but does not provide an independent mechanism for controlling cross-time semantic connectivity. Section~\ref{sec:timelift} introduces a score-preserving geometry adjustment that provides such control without changing any query--data TDVS score.

\subsection{Query-Orthogonal TimeLift}
\label{sec:timelift}

To independently adjust the influence of semantic similarity on data--data geometry without changing the TDVS objective, Chronos introduces a query-orthogonal geometry adjustment called \emph{TimeLift}. The key idea is to append an additional semantic component to each data representation, while appending a zero component to every query representation. The added component therefore affects data--data distances but contributes nothing to query--data scores.
Theorem~\ref{thm:timelift} (below) formalizes the key query-equivalence property of TimeLift: for any fixed $\kappa$, it changes the data--data geometry while preserving all query--data TDVS scores and rankings.

Let $\mathcal{H}_m$ denote the mode-specific Hilbert space defined in Section~\ref{sec:chronos_metricization}. For a geometry parameter $\kappa\geq0$, TimeLift extends this space to $\widehat{\mathcal{H}}_m=\mathcal{H}_m\oplus\mathbb{R}^d$.

\begin{definition}[Query-Orthogonal TimeLift]
\label{def:timelift}
For each TDVS mode $m\in\{\mathrm{mul},\mathrm{add}\}$, TimeLift
maps a query $(q,\tau)$ and a timestamped data vector $(x_i,t_i)$ to
\begin{equation}
    \widehat{Q}_m^{(\kappa)}(q,\tau)
    =
    \begin{bmatrix}
        Q_m(q,\tau)\\
        \mathbf{0}_d
    \end{bmatrix},
    \qquad
    \widehat{P}_m^{(\kappa)}(x_i,t_i)
    =
    \begin{bmatrix}
        P_m(x_i,t_i)\\
        \sqrt{\kappa}\,x_i
    \end{bmatrix},
    \label{eq:timelift-mapping}
\end{equation}
where $\mathbf{0}_d$ is the zero vector in $\mathbb{R}^d$.
\end{definition}

The appended semantic component $\sqrt{\kappa}x_i$ is
query-orthogonal because the corresponding component of $\widehat{Q}_m^{(\kappa)}(q,\tau)$ is $\mathbf{0}_d$.
Let $d_m^{(\kappa)}$ be the TimeLift metric and
$D_m^{(\kappa)}=(d_m^{(\kappa)})^2$ be its squared form.

\begin{theorem}[Query-Equivalent TimeLift]
\label{thm:timelift}
For every $\kappa\geq0$ and each TDVS mode
$m\in\{\mathrm{mul},\mathrm{add}\}$, TimeLift preserves all
query--data TDVS scores. Moreover, every mapped query has unit norm,
whereas all mapped data vectors have the same norm
$\sqrt{1+\kappa}$. For any fixed $\kappa$, nearest-neighbor search using either $d_m^{(\kappa)}$ or its squared form $D_m^{(\kappa)}$ preserves the exact TDVS ranking.
\end{theorem}

\begin{proof}
By Definition~\ref{def:timelift} and
Theorem~\ref{thm:chronos-metricization},
\begin{equation}
    \left\langle
        \widehat{Q}_m^{(\kappa)},
        \widehat{P}_m^{(\kappa)}
    \right\rangle
    =
    \langle Q_m,P_m\rangle
    +
    \left\langle
        \mathbf{0}_d,\sqrt{\kappa}x_i
    \right\rangle
    =
    F_m(q,x_i,\tau).
\end{equation}
Since $\|Q_m\|=\|P_m\|=\|x_i\|=1$, the mapped query and data norms
are $1$ and $\sqrt{1+\kappa}$, respectively. Hence,
\begin{equation}
    D_m^{(\kappa)}(q,x_i)
    =
    \left(d_m^{(\kappa)}(q,x_i)\right)^2
    =
    \left\|
    \widehat{Q}_m^{(\kappa)}
    -
    \widehat{P}_m^{(\kappa)}
    \right\|^2
    =
    2+\kappa-2F_m(q,x_i,\tau).
    \label{eq:chronos-query-data-distance-kappa}
\end{equation}
For any fixed $\kappa$, minimizing either $d_m^{(\kappa)}$ or its squared form $D_m^{(\kappa)}$ is exactly equivalent to maximizing the TDVS score.
\end{proof}

\paragraph{Controllable data geometry.}
Although $\kappa$ does not affect any query--data score, it changes the pairwise distances among indexed data vectors. Using $c_{ij}$ and $k_{ij}$ defined in Section~\ref{sec:chronos_metricization}, the squared TimeLift distance for multiplicative TDVS is
\begin{align}
    D_{\mathrm{mul}}^{(\kappa)}(x_i,x_j)
    &=
    \left\|
        \widehat{P}_{\mathrm{mul}}^{(\kappa)}(x_i,t_i)
        -
        \widehat{P}_{\mathrm{mul}}^{(\kappa)}(x_j,t_j)
    \right\|^2
    \nonumber\\
    &=
    2(1-k_{ij})
    +
    2(k_{ij}+\kappa)(1-c_{ij}).
    \label{eq:timelift-mul-dd}
\end{align}
For additive TDVS, the corresponding squared distance is
\begin{align}
    D_{\mathrm{add}}^{(\kappa)}(x_i,x_j)
    &=
    \left\|
        \widehat{P}_{\mathrm{add}}^{(\kappa)}(x_i,t_i)
        -
        \widehat{P}_{\mathrm{add}}^{(\kappa)}(x_j,t_j)
    \right\|^2
    \nonumber\\
    &=
    2(1-\alpha)(1-k_{ij})
    +
    2(\alpha+\kappa)(1-c_{ij}).
    \label{eq:timelift-add-dd}
\end{align}

Compared with the base data geometry, TimeLift introduces a timestamp-independent semantic term $2\kappa(1-c_{ij})$. In the multiplicative case, the coefficient of semantic separation increases from $k_{ij}$ to $k_{ij}+\kappa$. Thus, even when two timestamps are far apart and $k_{ij}$ approaches zero, their squared distance still depends on semantic similarity through $2\kappa(1-c_{ij})$. In the additive case, TimeLift similarly adjusts the semantic contribution from $\alpha$ to $\alpha+\kappa$ in the data geometry, without changing the weight $\alpha$ in the TDVS objective.

The parameter $\kappa$ controls the indexing geometry: $\kappa=0$ recovers the base geometry, while larger values strengthen the influence of semantic similarity and connectivity across distant timestamps.
Because all $\kappa$ values preserve the exact TDVS ranking while inducing distinct data--data geometries, Chronos defines a query-equivalent but geometry-distinct family.
However, when constructing an index, using one global $\kappa$ forces a trade-off between temporal locality and cross-time semantic connectivity. TANGO (Section~\ref{sec:tango}) addresses this trade-off by assigning layer-specific $\kappa_\ell$ values within a hierarchical graph while preserving the same TDVS ordering across all layers.

\subsection{Efficient Realization}
\label{sec:chronos_realization}

Chronos uses the Hilbert-space mappings only for theoretical construction and does not explicitly construct or store the temporal features $\phi_\lambda(t)$, the tensor-product or direct-sum representations $Q_m(q,\tau)$ and $P_m(x_i,t_i)$, or their TimeLift extensions $\widehat{Q}_m^{(\kappa)}(q,\tau)$ and $\widehat{P}_m^{(\kappa)}(x_i,t_i)$.
This is because Chronos derives exact formulas for the pairwise distance required by query processing (query--data) and index construction (data--data) in the mapped Hilbert spaces, allowing the score-preserving and geometry-controllable Hilbert-space geometry to be realized directly from the original embeddings and timestamps without materializing the mapped representations.

Specifically, during query processing, Chronos directly evaluates the original TDVS score $F_m(q,x_i,\tau)$. Equivalently, for any fixed $\kappa$, the same candidate ordering can be expressed through the query--data squared distance $D_m^{(\kappa)}(q,x_i)$ in Equation~\eqref{eq:chronos-query-data-distance-kappa}.
During index construction, Chronos computes data--data squared distances using $D_{\mathrm{mul}}^{(\kappa)}(x_i,x_j)$ and $D_{\mathrm{add}}^{(\kappa)}(x_i,x_j)$ in Equations~\eqref{eq:timelift-mul-dd}--\eqref{eq:timelift-add-dd}.

The TANGO graph index introduced in Section~\ref{sec:tango} uses these data--data squared distances to construct its graph neighborhoods, while using the exact TDVS score to guide query traversal.

\begin{figure}[t]
    \centering
    \includegraphics[width=0.88\linewidth]
    {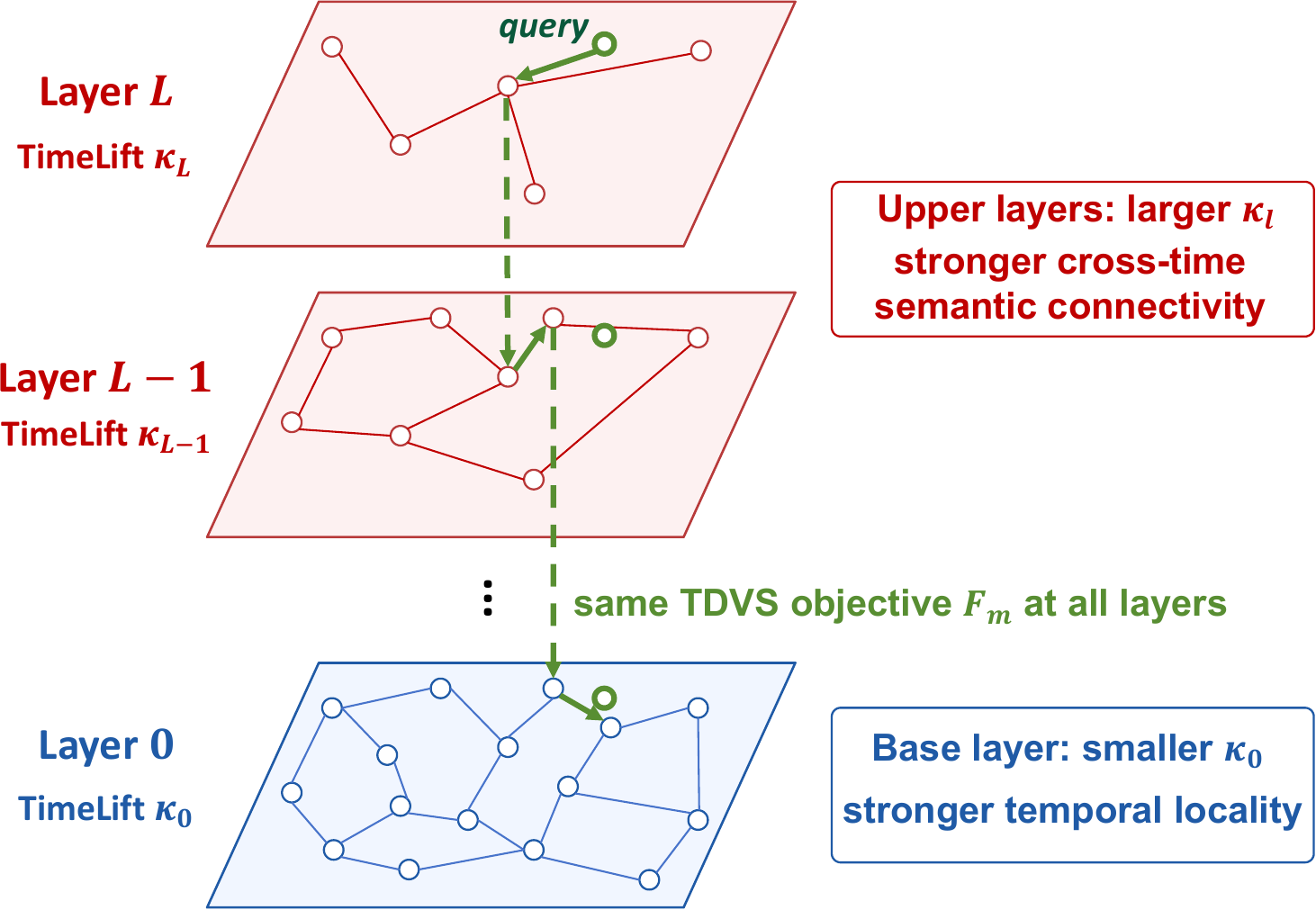}
    \caption{TANGO Hierarchical Index.}
    \label{fig:tango-hierarchy}
\end{figure}

\section{TANGO: A TDVS-Native Graph Index}
\label{sec:tango}

\paragraph{Overview.}
Chronos exposes a family of query-equivalent but geometry-distinct metric spaces parameterized by $\kappa$. However, a single value of $\kappa$ defines a fixed geometry that is difficult to simultaneously support long-range navigation and
local candidate refinement.
In this section, we propose \emph{TANGO}, short for \emph{\underline{T}ime-\underline{A}ware \underline{N}avigable \underline{G}raph with Query-\underline{O}rthogonal TimeLift}, a TDVS-native hierarchical graph index that realizes different Chronos geometries across graph layers while preserving a common query objective.

TANGO follows the hierarchical organization of HNSW~\cite{malkov2018efficient}.
As illustrated in Figure~\ref{fig:tango-hierarchy}, TANGO's key distinction from HNSW is that each layer $\ell$ is constructed under a distinct TimeLift geometry parameterized by $\kappa_\ell$: larger values in upper layers strengthen cross-time semantic connectivity, while the smaller $\kappa_0$ at the base layer preserves temporal locality. 
This layer specialization allows TANGO to navigate rapidly across temporally distant regions in the upper layers and then quickly focus on locally relevant semantic--temporal neighborhoods at the base layer. 
Importantly, layer-specific geometries introduce no objective mismatch: $\kappa_\ell$ affects only the data--data geometry used for graph construction, while query--data ordering remains exactly the TDVS ordering. 
Consequently, all layers can be searched using the same exact TDVS score $F_m(q,x_i,\tau)$ within a single top-down traversal.
Moreover, with temporal factors cached and shared across layers, TANGO avoids repeated exponential computations for temporal decay and requires only one semantic inner product and constant-time scalar operations per visited candidate, enabling efficient query processing.

\subsection{Layer-Wise Graph Geometry}
\label{sec:tango_geometry}

Let
$
    \mathcal{G}
    =
    \left\{
        \mathcal{G}_{\ell}
        =
        (\mathcal{V}_{\ell},\mathcal{E}_{\ell})
    \right\}_{\ell=0}^{L}
$
denote the TANGO hierarchy, where $\ell=0$ is the base layer and $\mathcal{V}_{\ell+1}\subseteq\mathcal{V}_{\ell}$. As in HNSW, each data vector is assigned a maximum level, and higher layers contain progressively fewer vectors. TANGO assigns a separate nonnegative TimeLift parameter to each graph layer:
\begin{equation}
    \boldsymbol{\kappa}
    =
    (\kappa_0,\kappa_1,\ldots,\kappa_L),
    \qquad
    \kappa_\ell\geq0.
    \label{eq:tango-layerwise-kappa}
\end{equation}
A larger $\kappa_\ell$ gives semantic similarity greater influence over the data geometry of layer $\ell$, whereas a smaller $\kappa_\ell$ preserves stronger temporal locality.
Within TANGO's multi-layer hierarchy, sparse upper layers benefit from stronger cross-time semantic
connectivity for long-range navigation, whereas lower layers benefit from greater temporal locality for fine-grained refinement.
Thus, TANGO adopts a nondecreasing schedule:
$0\leq\kappa_0\leq\kappa_1\leq\cdots\leq\kappa_L.$

From Equation~\eqref{eq:chronos-query-data-distance-kappa}, $D_m^{(\kappa_\ell)}(q,x_i)
=2+\kappa_\ell-2F_m(q,x_i,\tau)$. Hence, different values of $\kappa_\ell$ add different layer-dependent
constants to the squared query--data distances, so their numerical values are not directly comparable across layers. Nevertheless, the candidate ordering remains identical across all layers. TANGO therefore uses the common TDVS score $F_m$ throughout query traversal, while $\kappa_\ell$ determines only the graph neighborhoods
constructed at layer $\ell$.

At layer $\ell$, all data--data comparisons are based on $D_m^{(\kappa_\ell)}(x_i,x_j)$, i.e., the multiplicative or additive TimeLift distance defined in Equations~\eqref{eq:timelift-mul-dd} and~\eqref{eq:timelift-add-dd}, respectively. Since graph construction depends only on distance ordering, the squared distance can be used without computing its square root.

\begin{algorithm}[tb]
    % \small
    \scriptsize
    \caption{TANGO Layer Search}
    \label{alg:tango-layer-search}
    \LinesNumbered
    \KwIn{Layer $\mathcal G_\ell$, TDVS mode
$m\in\{\mathrm{mul},\mathrm{add}\}$, search budget $ef$, cached data factors $\{b_i\}$, phase $r\in\{\mathrm{build},\mathrm{query}\}$, entry set $\mathcal E$;
    vector to be inserted $v=(x_v,t_v)$ and layer parameter
    $\kappa_\ell$ if $r=\mathrm{build}$;
    query $(q,\tau)$ and query factor $a_\tau$ if
    $r=\mathrm{query}$}
    \KwOut{Candidate set $\mathcal W$}

    \eIf{$r=\mathrm{build}$}{
        For any vertex $u\in\mathcal V_\ell$, define the smaller-is-better key
        $\delta_r(u)\leftarrow
        D_m^{(\kappa_\ell)}(x_v,x_u)$, where
        $k_{vu}\leftarrow
        \frac{\min\{b_v,b_u\}}{\max\{b_v,b_u\}}$\;
    }{
        For any vertex $u\in\mathcal V_\ell$, define the
        smaller-is-better key
        $\delta_r(u)\leftarrow
        -F_m(q,x_u,\tau)$, where $e^{-\lambda(\tau-t_u)} \leftarrow a_\tau b_u$\;
    }

    Initialize $\mathcal C$ (a min-priority queue of unexpanded candidates) and $\mathcal W$ (a max-priority queue of the current best vertices) with $\mathcal E$, both ordered by $\delta_r$; set the visited set $\mathcal Z\leftarrow\mathcal E$\;

    \While{$\mathcal C\neq\emptyset$}{
        $c\leftarrow$ remove the vertex with the smallest
        $\delta_r$ from $\mathcal C$\;
        $w\leftarrow$ the vertex with the largest
        $\delta_r$ in $\mathcal W$\;

        \If{$|\mathcal W|\geq ef$ and
        $\delta_r(c)>\delta_r(w)$}{
            \textbf{break}\;
        }

        \ForEach{$u\in\mathcal N_\ell(c)$}{
            \If{$u\notin\mathcal Z$}{
                $\mathcal Z\leftarrow\mathcal Z\cup\{u\}$\;
                Compute $\delta_r(u)$ using the
                phase-specific key\;

                \If{$|\mathcal W|<ef$ or
$\delta_r(u)<\max_{z\in\mathcal W}\delta_r(z)$}{
                    Insert $u$ into $\mathcal C$ and
                    $\mathcal W$\;

                    \If{$|\mathcal W|>ef$}{
                        Remove the largest-$\delta_r$ vertex from $\mathcal W$\;
                    }
                }
            }
        }
    }

    \Return{$\mathcal W$}\;
\end{algorithm}

\paragraph{Layer-shared temporal factor caching.}
During index construction and query traversal, TANGO would repeatedly compute the data--data temporal affinity $e^{-\lambda|t_i-t_j|}$ and the query--data decay $e^{-\lambda(\tau-t_i)}$, causing nontrivial computational overhead.
Since these terms depend only on timestamps and the fixed decay rate $\lambda$, they can be factorized into per-data and per-query temporal factors and cached for reuse. 
Moreover, these cached factors are independent of the layer-specific parameter $\kappa_\ell$, allowing the same values to be shared across all graph layers.

Let $t_0$ be a reference timestamp chosen for numerical stability; TANGO sets it to the latest timestamp when the index is initialized. It caches one per-data temporal factor $b_i$ for each indexed vector and computes one per-query temporal factor $a_\tau$ for each query:
\begin{equation}
    b_i
    =
    e^{\lambda(t_i-t_0)},
    \qquad
    a_\tau
    =
    e^{-\lambda(\tau-t_0)}.
    \label{eq:tango-cached-temporal-factors}
\end{equation}
The required query--data decay and data--data temporal affinity are then computed exactly as
\begin{equation}
    e^{-\lambda(\tau-t_i)}
    =
    a_\tau b_i,
    \qquad
    e^{-\lambda|t_i-t_j|}
    =
    \frac{
        \min\{b_i,b_j\}
    }{
        \max\{b_i,b_j\}
    }.
    \label{eq:tango-cached-temporal-evaluation}
\end{equation}
Caching these temporal factors eliminates repeated exponentiation from pairwise comparisons: TANGO performs one exponential computation per inserted vector and one per query, while subsequent temporal terms are obtained using only scalar multiplication or division and are shared across all graph layers. 
Unlike the temporal anchor $T$ in STR, $t_0$ is only a numerical reference and does not alter any retrieval score, data geometry, or graph edge. 
As time advances, although keeping $t_0$ fixed does not affect correctness, the numerical scale of the cached factors may gradually drift.
TANGO can then replace $t_0$ with a newer reference $t_0'$ and uniformly rescale each cached factor as
$b_i\leftarrow e^{-\lambda(t_0'-t_0)}b_i$, without recomputing pairwise distances or rebuilding the graph.

\begin{algorithm}[tb]
    % \small
    \scriptsize
    \caption{Index Construction and Online Insertion}
    \label{alg:tango-insertion}
    \LinesNumbered
    \KwIn{TANGO hierarchy
    $\mathcal G=\{\mathcal G_\ell\}_{\ell=0}^{L}$, layer parameters $\boldsymbol{\kappa}$,
    TDVS mode $m\in\{\mathrm{mul},\mathrm{add}\}$,
    vector to be inserted $v=(x_v,t_v)$, decay rate $\lambda$, reference timestamp $t_0$,
    cached data factors $\{b_i\}$, 
     budget $\mathit{efConstruction}$,
    and degree limits $\{M_\ell\}$}
    \KwOut{Updated hierarchy $\mathcal G$ and cached factors
    $\{b_i\}$}

    Compute and cache
    $b_v\leftarrow e^{\lambda(t_v-t_0)}$
    for reuse across all layers\;
    Sample the maximum level $L_v$ using the HNSW level
    distribution\;

    \If{$\mathcal G$ is empty}{
        \For{$\ell=0$ to $L_v$}{
            Add $v$ to $\mathcal V_\ell$\;
        }
        Set $v$ as the global entry point and
        $L\leftarrow L_v$\;
        \Return{$(\mathcal G,\{b_i\})$}\;
    }

    Set the entry set
    $\mathcal E\leftarrow\{\text{global entry point}\}$\;

    \For{$\ell=L$ down to $L_v+1$}{
        $\mathcal E\leftarrow$
        \textbf{call}
        \emph{TANGO Layer Search}
        $(\mathcal G_\ell,m,1,\{b_i\},
        \mathrm{build},\mathcal E,v,\kappa_\ell)$\;
    }

    \For{$\ell=\min\{L,L_v\}$ down to $0$}{
        $\mathcal W_\ell\leftarrow$
        \textbf{call}
        \emph{TANGO Layer Search}
        $(\mathcal G_\ell,m,\mathit{efConstruction},
        \{b_i\},\mathrm{build},\mathcal E,
        v,\kappa_\ell)$\;

        $\mathcal S_\ell\leftarrow$
        \textbf{call}
        \emph{Diversified Neighbor Selection}
        $(v,\mathcal W_\ell,M_\ell,
        D_m^{(\kappa_\ell)}(\cdot,\cdot),\{b_i\})$\;

        Add $v$ to $\mathcal V_\ell$ and set
        $\mathcal N_\ell(v)\leftarrow\mathcal S_\ell$\;

        \ForEach{$u\in\mathcal S_\ell$}{
            Add the reverse edge
            $\mathcal N_\ell(u)\leftarrow
            \mathcal N_\ell(u)\cup\{v\}$\;

            \If{$|\mathcal N_\ell(u)|>M_\ell$}{
                $\mathcal N_\ell(u)\leftarrow$
                \textbf{call}
                \emph{Diversified Neighbor Selection}
                $(u,\mathcal N_\ell(u),M_\ell,
                D_m^{(\kappa_\ell)}(\cdot,\cdot),\{b_i\})$\;
            }
        }

        $\mathcal E\leftarrow\mathcal W_\ell$\;
    }

    \If{$L_v>L$}{
        \For{$\ell=L+1$ to $L_v$}{
            Add $v$ to $\mathcal V_\ell$\;
        }
        Set $v$ as the global entry point and
        $L\leftarrow L_v$\;
    }

    \Return{$(\mathcal G,\{b_i\})$}\;
\end{algorithm}

\subsection{Index Construction and Online Insertion}
\label{sec:tango_construction}

TANGO retains the hierarchical organization, level assignment, and bounded-degree structure of HNSW, but refines the proximity computations used to construct the graph.
Specifically, every distance-based decision at layer $\ell$ uses the layer-specific squared TimeLift distance
$D_m^{(\kappa_\ell)}(x_i,x_j)$ (Equations~\eqref{eq:timelift-mul-dd}--\eqref{eq:timelift-add-dd}). The temporal affinity required by this distance is computed from the cached factors as
$
    k_{ij}=e^{-\lambda|t_i-t_j|}
    =
    \frac{\min\{b_i,b_j\}}{\max\{b_i,b_j\}}.
$
Consequently, TANGO retains HNSW's efficient hierarchical navigation while natively incorporating the semantic--temporal geometry of TDVS, thereby supporting efficient index construction and query processing.

\paragraph{TANGO layer search.}
Algorithm~\ref{alg:tango-layer-search} presents the common layer-search procedure used during both construction and query processing. The two phases share the same best-first queue operations but use different ranking keys. During construction, vertices at layer $\ell$ are ordered by the layer-specific squared TimeLift distance $D_m^{(\kappa_\ell)}(x_v,x_u)$, whose temporal affinity is obtained from $b_v$ and $b_u$. During query processing, vertices are instead ordered by the negative exact TDVS score
$-F_m(q,x_u,\tau)$, with the query--data decay evaluated as $a_\tau b_u$.

The shared search procedure therefore does not impose a common proximity measure on the two phases. The construction key determines the neighborhoods and topology of each layer, whereas the query key directly follows the target TDVS objective. In both phases, temporal factor caching eliminates repeated exponentiation during graph traversal.

\paragraph{Index construction.}
Algorithm~\ref{alg:tango-insertion} presents the complete index construction procedure.
To insert a timestamped vector $v=(x_v,t_v)$, TANGO computes and caches its temporal factor $b_v=e^{\lambda(t_v-t_0)}$ once and reuses it across all graph layers.
TANGO retains HNSW's hierarchical insertion and diversified neighbor selection mechanisms, but replaces their underlying proximity measure with the TimeLift geometry of the current layer.
During the top-down traversal and candidate exploration, every data--data comparison at layer $\ell$ is evaluated using the layer-specific squared TimeLift distance $D_m^{(\kappa_\ell)}(x_i,x_j)$, with the temporal affinity obtained directly from the cached factors.

% TANGO's construction differs from standard HNSW in three key aspects. First, each layer is constructed under its own geometry parameter $\kappa_\ell$, enabling layer-wise semantic--temporal neighborhoods. Second, all distance-based construction decisions within a layer consistently use the same squared TimeLift distance $D_m^{(\kappa_\ell)}$. Third, its temporal affinity is evaluated from the cached factors $b_i$ and $b_j$, eliminating repeated exponentiation from data--data comparisons.

\paragraph{Online insertion.}
The same procedure directly supports online insertion. For each newly arriving vector, TANGO computes its cached temporal factor $b_v$ and performs the corresponding local searches, edge insertions, and adjacency-list pruning. Since
$D_m^{(\kappa_\ell)}(x_i,x_j)$ depends on the relative timestamp difference $|t_i-t_j|$, inserting a vector leaves existing data--data distances unchanged and requires \emph{no} global updates to previously indexed vectors. 
TANGO remains temporally stable under subsequent insertions.

\begin{algorithm}[tb]
    % \small
    \scriptsize
    \caption{Query Processing}
    \label{alg:tango-query}
    \LinesNumbered
    \KwIn{TANGO hierarchy
    $\mathcal G=\{\mathcal G_\ell\}_{\ell=0}^{L}$, query $(q,\tau)$, decay rate $\lambda$,
    TDVS mode $m\in\{\mathrm{mul},\mathrm{add}\}$, 
    reference timestamp $t_0$, cached data factors $\{b_i\}$, result size $k$, budget $\mathit{efSearch}\geq k$}
    \KwOut{Approximate TDVS top-$k$ result
    $\widehat{\mathcal N}_k$}

    Compute the query factor
    $a_\tau\leftarrow e^{-\lambda(\tau-t_0)}$
    for reuse across layers\;
    Set the entry set
    $\mathcal E\leftarrow\{\text{global entry point}\}$\;

    \For{$\ell=L$ down to $1$}{
        $\mathcal E\leftarrow$
        \textbf{call}
        \emph{TANGO Layer Search}
        $(\mathcal G_\ell,m,1,\{b_i\},
        \mathrm{query},\mathcal E,
        (q,\tau),a_\tau)$\;
    }

    $\mathcal W_0\leftarrow$
    \textbf{call}
    \emph{TANGO Layer Search}
    $(\mathcal G_0,m,\mathit{efSearch},
    \{b_i\},\mathrm{query},\mathcal E,
    (q,\tau),a_\tau)$\;

    $\widehat{\mathcal N}_k
    \leftarrow
    \operatorname{TopK}_{u\in\mathcal W_0}^{k}
    F_m(q,x_u,\tau)$\;

    \Return{$\widehat{\mathcal N}_k$}\;
\end{algorithm}

\subsection{Query Processing}
\label{sec:tango_query}

The layer-specific parameters $\kappa_\ell$ affect the query result through the graph topology constructed at each layer, but they are deliberately absent from the query-time scoring function. As shown in Section~\ref{sec:tango_geometry}, changing $\kappa_\ell$ adds the same constant to all squared
query--data distances within layer $\ell$ and therefore does not change candidate ordering. TANGO consequently traverses every layer using the same exact TDVS score $F_m(q,x_i,\tau)$.

Algorithm~\ref{alg:tango-query} presents the query processing procedure. TANGO computes the query factor
$a_\tau=e^{-\lambda(\tau-t_0)}$ once and reuses it across graph layers during hierarchical traversal. For each visited vertex $u$, the query--data decay is obtained as $a_\tau b_u$, and the smaller-is-better key is defined as $\delta_{\mathrm{query}}(u)=-F_m(q,x_u,\tau)$ in Algorithm~\ref{alg:tango-layer-search}. 
By invoking Algorithm~\ref{alg:tango-layer-search}, TANGO performs upper-layer navigation and base-layer exploration under the same exact TDVS scoring rule, and returns the $k$ candidates with the largest TDVS scores.

Taken together, TANGO assigns distinct roles to graph geometry and query scoring.
The layer-specific graph geometries determine the neighborhoods and connectivity available for traversal, while the exact TDVS score determines the exploration priority of visited vertices and their ranking.
With temporal factor caching, TANGO supports TDVS via one hierarchical traversal without materializing Chronos representations, requiring only one semantic inner product and constant-time scalar operations per visited candidate.

\section{Experimental Evaluation}
\label{sec:experiments}

In this section, we conduct experiments to evaluate TANGO. Our evaluation addresses three questions: 
(1) how TANGO compares to state-of-the-art graph-based methods in query and index performance (Section~\ref{sec:overall_performance}); 
(2) how robust TANGO is across different temporal settings (Section~\ref{sec:robustness});
and (3) how efficiently TANGO supports online insertion (Section~\ref{sec:insertion}).

All experiments are conducted on a server equipped with two Intel Xeon Platinum 8458P processors, providing 88 physical cores and 881 GiB of memory. The server runs Ubuntu 22.04.5 with Linux 5.15. All methods use 64 threads for index construction and a single thread for query processing. All implementations are compiled with GCC 11.4 in C++17 release mode using \texttt{-O3} and \texttt{-march=native}.

\subsection{Experimental Setup}
\label{sec:experimental-setup}

\noindent \textbf{Datasets and queries.}
We evaluate all methods on seven datasets of native unit-length embeddings, covering cardinalities from 500K to 10M and dimensionalities from 96 to 4096. Table~\ref{tab:datasets} summarizes their characteristics.
For Wikipedia-Qwen~\cite{zhu2025wikipedia_qwen}, OpenAI-1536~\cite{qdrant_openai1536}, and OpenAI-3072~\cite{qdrant_openai3072}, we hold out 1K vectors from each original collection as a strictly disjoint query set.
MSMARCO-1M and MSMARCO-10M are randomly sampled from the 113.5M-passage collection and use its independent real queries~\cite{cohere_msmarco_v21}.
HotpotQA uses its complete corpus and official queries~\cite{cohere_beir_embed}, while Yandex Deep-10M uses the first 10M Deep1B vectors and its public queries~\cite{bigann2021}.
Every dataset is evaluated using exactly 1,000 queries.

% \noindent \textbf{Datasets and queries.}
% We evaluate all methods on seven datasets of native unit-length embeddings, covering cardinalities from 500K to 10M and dimensionalities from 96 to 4096. Table~\ref{tab:datasets} summarizes their characteristics. 
% For Wikipedia-Qwen\footnote{https://huggingface.co/datasets/maknee/wikipedia\_qwen\_8b}, OpenAI-1536\footnote{https://huggingface.co/datasets/Qdrant/dbpedia-entities-openai3-text-embedding-3-large-1536-1M}, and OpenAI-3072\footnote{https://huggingface.co/datasets/Qdrant/dbpedia-entities-openai3-text-embedding-3-large-3072-1M}, we hold out 1K vectors from each original collection as a strictly disjoint query set. MSMARCO-1M and MSMARCO-10M are randomly sampled from the 113.5M-passage collection and use its independent real queries\footnote{https://huggingface.co/datasets/CohereLabs/msmarco-v2.1-embed-english-v3}. HotpotQA uses its complete corpus and official queries\footnote{https://huggingface.co/datasets/CohereLabs/beir-embed-english-v3}, while Yandex Deep-10M uses the first 10M Deep1B vectors and its public queries\footnote{https://big-ann-benchmarks.com/neurips21.html}.
% Every dataset is evaluated using exactly 1,000 queries. 

\begin{table}[tb]
  \centering
  \caption{Dataset statistics.}
  \label{tab:datasets}
  \small
  \setlength{\tabcolsep}{5pt}
  \renewcommand{\arraystretch}{1.05}
  \begin{tabular}{lcccc}
    \toprule
    \textbf{Dataset} &
    \textbf{Data Size} &
    \textbf{Dim.} &
    \textbf{Query Size} &
    \textbf{Data Type} \\
    \midrule
    Wikipedia-Qwen & 500{,}000       & 4096 & 1{,}000 & Text  \\
    OpenAI-1536     & 1{,}000{,}000  & 1536 & 1{,}000 & Text  \\
    OpenAI-3072     & 1{,}000{,}000  & 3072 & 1{,}000 & Text  \\
    MSMARCO-1M      & 1{,}000{,}000  & 1024 & 1{,}000 & Text  \\
    HotpotQA        & 5{,}233{,}329  & 1024 & 1{,}000 & Text  \\
    Yandex Deep-10M & 10{,}000{,}000 & 96   & 1{,}000 & Image \\
    MSMARCO-10M     & 10{,}000{,}000 & 1024 & 1{,}000 & Text  \\
    \bottomrule
  \end{tabular}
\end{table}

\begin{figure*}[tb]
    \centering

    \begin{minipage}[t]{0.59\textwidth}
        \centering
        \includegraphics[width=\linewidth]
        {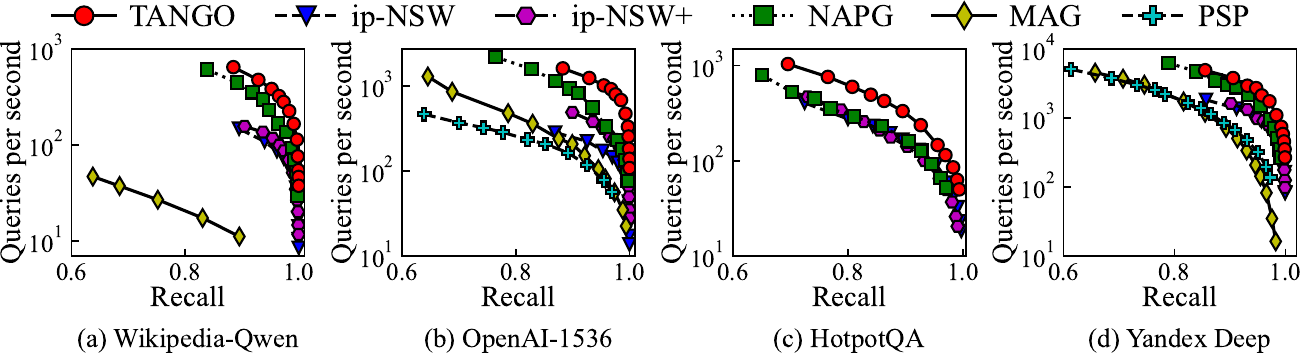}
        \captionof{figure}{
            Query performance comparison for multiplicative TDVS.
        }
        \label{fig:query_multiplicative}
    \end{minipage}
    \hspace{0.025\textwidth}
    \begin{minipage}[t]{0.36\textwidth}
        \centering
        \includegraphics[width=\linewidth]
        {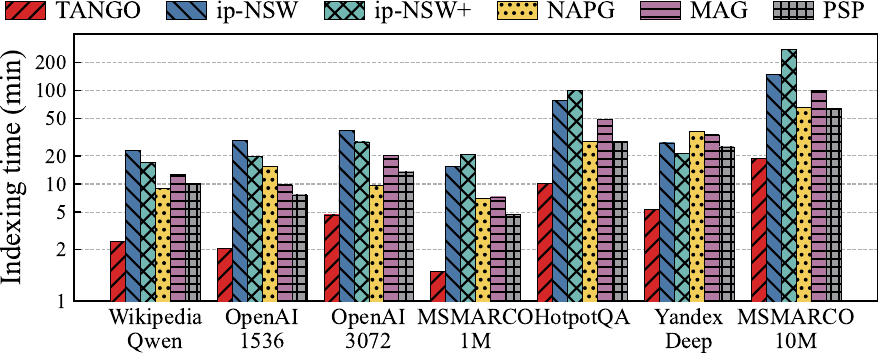}
        \captionof{figure}{
            Indexing time.
        }
        \label{fig:index-construction-time}
    \end{minipage}

\end{figure*}

\begin{figure*}[tb]
    \centering

    \begin{minipage}[t]{0.59\textwidth}
        \centering
        \includegraphics[width=\linewidth]
        {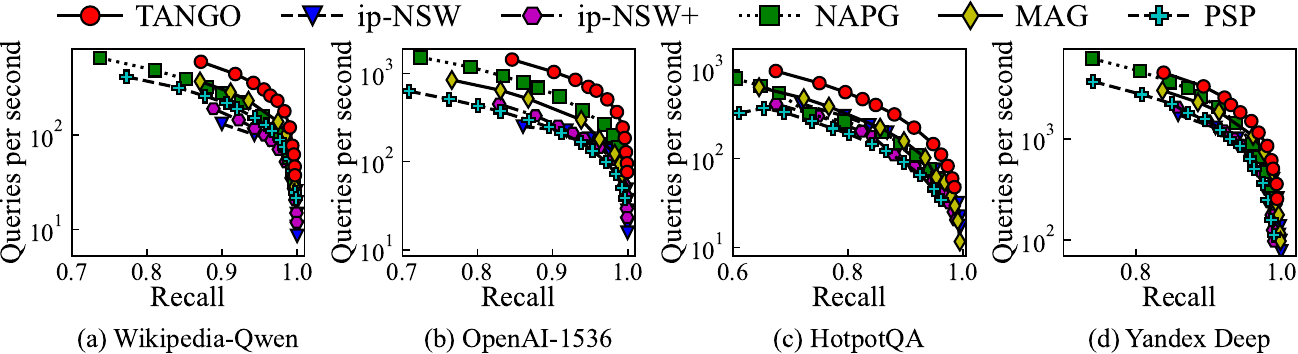}
        \captionof{figure}{
            Query performance comparison for additive TDVS.
        }
        \label{fig:query_additive}
    \end{minipage}
    \hspace{0.025\textwidth}
    \begin{minipage}[t]{0.36\textwidth}
        \centering
        \includegraphics[width=\linewidth]
        {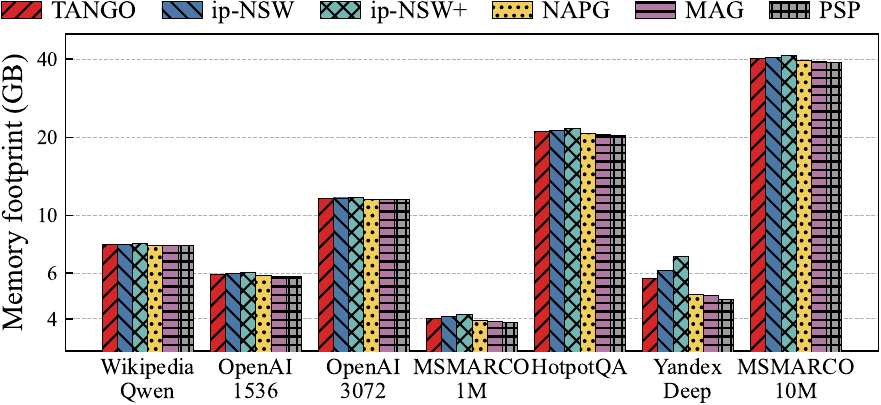}
        \captionof{figure}{
            Index memory footprint.
        }
        \label{fig:index-footprint}
    \end{minipage}

\end{figure*}

\noindent \textbf{Temporal workloads.}
Since there is no public dataset with native timestamps, we assign each base vector $x_i$ a timestamp $t_i\in[0,H]$, with larger values denoting newer objects, and model two representative semantic--temporal relationships. 
The \emph{topic-independent} workload captures scenarios where timestamps are unrelated to semantic
content, such as general-purpose document repositories and continuously collected logs. We construct this workload by drawing each timestamp independently from $\mathrm{Uniform}(0,H)$. 
The \emph{topic-correlated} workload captures scenarios where semantically related content occurs during similar periods, such as news articles and conversations with large language models, with topics remaining active over limited time windows. 
We construct this workload by clustering the unit vectors using spherical $k$-means, assigning each cluster $c$ an activity center $\mu_c\sim\mathrm{Uniform}(0,H)$, and sampling its member timestamps from a normal distribution centered at $\mu_c$ and truncated to $[0,H]$.

Following common practice, we parameterize exponential decay using the interpretable half-life $h$, where $\lambda=\ln 2/h$, and use $\alpha$ to control the semantic--temporal trade-off in additive TDVS. Unless otherwise stated, we set $H=\tau=365$ days, $h=90$ days, $\alpha=0.7$, adopt the topic-correlated workload, and evaluate both multiplicative and additive TDVS. Subsequent experiments vary the workload type, $h$, and $\alpha$ to comprehensively evaluate TANGO's robustness.

\noindent \textbf{Benchmark methods.}
This paper presents two solutions to TDVS: adapting state-of-the-art MIPS indexes through STR (Section~\ref{sec:str}) and natively indexing the Chronos geometry with TANGO (Sections~\ref{sec:chronos_framework} and~\ref{sec:tango}). 
Accordingly, we compare TANGO with five state-of-the-art graph-based MIPS methods adapted to TDVS via STR~\cite{morozov2018non,liu2020understanding,tan2021norm,chen2025maximum,chen2025stitching}. 
These methods span major advances in graph-based MIPS: \textbf{ip-NSW}~\cite{morozov2018non} directly constructs a non-metric
inner-product graph; \textbf{ip-NSW+}~\cite{liu2020understanding} introduces angular navigation; \textbf{NAPG}~\cite{tan2021norm} adopts norm-adjusted graph construction; \textbf{MAG}~\cite{chen2025stitching} combines inner-product and Euclidean geometries; and \textbf{PSP}~\cite{chen2025maximum} exploits query-scaled nearest-neighbor geometry with spherical pathways. 
Together, they cover foundational and recent state-of-the-art designs, forming strong baselines for both multiplicative and additive TDVS.

\noindent \textbf{Evaluation measures.}
For each query $(q,\tau)\in\mathcal Q$, let $\widehat{\mathcal N}^{\mathrm{tdvs}}_k(q,\tau)$ and $\mathcal N^{\mathrm{tdvs}}_k(q,\tau)$ denote the returned and exact top-$k$.
We measure answer quality by
$
    \mathrm{Recall@}k=\frac{1}{|\mathcal Q|}\sum_{(q,\tau)\in\mathcal Q}\frac{\left|\widehat{\mathcal N}^{\mathrm{tdvs}}_k(q,\tau) \cap \mathcal N^{\mathrm{tdvs}}_k(q,\tau) \right|}{k},
$
and query efficiency by queries per second (QPS). We sweep the search budget to obtain complete Recall--QPS curves. 
For index evaluation, we report index construction time to measure indexing efficiency and index memory footprint to quantify storage overhead. The footprint includes all artifacts retained in the deployed index, including transformed vectors, but excludes temporary construction artifacts.

\noindent \textbf{Parameter settings.}
We use the authors' open-source implementations of
ip-NSW~\cite{ipnsw_code}, ip-NSW+~\cite{ipnswplus_code}, MAG~\cite{mag_code}, and
PSP~\cite{psp_code}; since NAPG has no public implementation, we implement it following the original paper. 
For all methods, we adopt the recommended parameters
from their papers or the default settings of their codebases.
Specifically, ip-NSW uses $M=32$ and $\mathit{efConstruction}=1024$; ip-NSW+ additionally builds an auxiliary cosine graph with $M=10$ and
$\mathit{efConstruction}=100$; NAPG uses $M=16$,
$\mathit{efConstruction}=100$, five norm ranges, and 100 samples per range; MAG uses $L=60$, $R=48$, $C=300$, $R_{\mathrm{IP}}=20$, $M=64$, and threshold $8$; and PSP uses $K=400$, $L=800$, $R=40$, angle $60^\circ$, $M=5$. 
Both MAG and PSP require an external $k$NN graph, which we construct using FAISS NN-Descent with $K=400$, $L=420$, and 12 iterations.
TANGO uses $M=25$, $\mathit{efConstruction}=200$, $\kappa_{\mathrm{base}}=\frac{1}{63}$ at the base layer, and $\kappa_{\mathrm{nav}}=\frac{1}{15}$ at all upper layers across all datasets and both TDVS modes.
We vary each method's search parameter over a sufficiently broad range and compare complete Recall--QPS curves.
The number of returned results $k$ is set to 50 by default.

\subsection{Overall Performance}
\label{sec:overall_performance}

\subsubsection{Query Performance}
\label{sec:query_performance}

Figures~\ref{fig:query_multiplicative} and~\ref{fig:query_additive} present the Recall--QPS trade-offs of all methods for multiplicative and additive TDVS. The four representative datasets cover dimensionalities from 96 to 4096, and dataset sizes from 500K to 10M. 
TANGO consistently achieves the best query performance under both TDVS modes across all datasets. 
Compared with the best-performing competitor, TANGO achieves speedups of up to $3.06\times$ (recall $0.95$) and $3.50\times$ (recall $0.99$) under multiplicative TDVS; the corresponding gains under additive TDVS are $2.41\times$ and $2.44\times$.
MAG and PSP fail to reach high recall on Wikipedia-Qwen and HotpotQA under multiplicative TDVS because dispersed true neighbors weaken local connectivity in their semantic Euclidean graphs, limiting navigation and pruning.
The results show that TANGO's advantages are more pronounced on high-dimensional and large-scale datasets, with the performance gap over the baselines often widening in the high-recall region. This suggests that TANGO's TDVS-native graph geometry is especially beneficial when high dimensionality, large scale, or stringent recall makes graph navigation more challenging.

% \begin{figure*}[tb] 
% %\vspace*{-1.0cm}
% 	\centering
% 	\includegraphics[width=0.85\linewidth]{./figures/main_query_multiplicative_k50.pdf}
%  % \vspace*{-0.3cm}
% 	\caption{Query performance comparison of all methods for multiplicative TDVS.}
% 	\label{fig:query_multiplicative}
%  % \vspace*{-0.1cm}
% \end{figure*}

% \begin{figure*}[tb] 
% % \vspace*{-0.3cm}
% 	\centering
% 	\includegraphics[width=0.85\linewidth]{./figures/main_query_additive_k50.pdf}
%  % \vspace*{-0.3cm}
% 	\caption{Query performance comparison of all methods for additive TDVS.}
% 	\label{fig:query_additive}
%  % \vspace*{-0.1cm}
% \end{figure*}

\subsubsection{Index Performance}
\label{sec:index_performance}

Figures~\ref{fig:index-construction-time} and
\ref{fig:index-footprint} report the index construction time and index memory footprint under multiplicative TDVS across all seven datasets; additive TDVS exhibits similar trends, so we present only the multiplicative setting to save space.
TANGO is the fastest to construct on every dataset, outperforming the fastest competitor by up to $4.05\times$.
This advantage stems from directly encoding temporal awareness into a compact unified graph, avoiding the expensive candidate expansion, auxiliary graphs, and $k$NN-graph-based multi-stage pruning and refinement required by the competitors.
All methods have similar memory footprints because they all employ sparse graph structures with comparable connectivity budgets, leaving only limited differences in method-specific auxiliary storage.
Overall, TANGO delivers substantially faster index construction while maintaining a comparable memory footprint, demonstrating superior overall index performance.

% \begin{figure}[tb] 
% %\vspace*{-0.7cm}
% 	\centering
% 	\includegraphics[width=\linewidth]{./figures/index_construction_time.pdf}
%     \vspace*{-0.6cm}
% 	\caption{Indexing time comparison.}
% 	\label{fig:index-construction-time}
% \end{figure}

% \begin{figure}[tb] 
% % \vspace*{-0.3cm}
% 	\centering
% 	\includegraphics[width=\linewidth]{./figures/index_memory_footprint.pdf}
%  \vspace*{-0.6cm}
% 	\caption{Index memory footprint comparison.}
% 	\label{fig:index-footprint}
% % \vspace*{-0.2cm}
% \end{figure}

\subsection{Robustness to Temporal Settings}
\label{sec:robustness}

We next examine TANGO's robustness to temporal settings by varying the timestamp distribution and half-life $h$ under both TDVS modes, and the semantic--temporal weight $\alpha$ under additive TDVS.

\begin{figure}[t]
%\vspace*{-0.7cm}
  \centering
  \includegraphics[width=\linewidth]
  {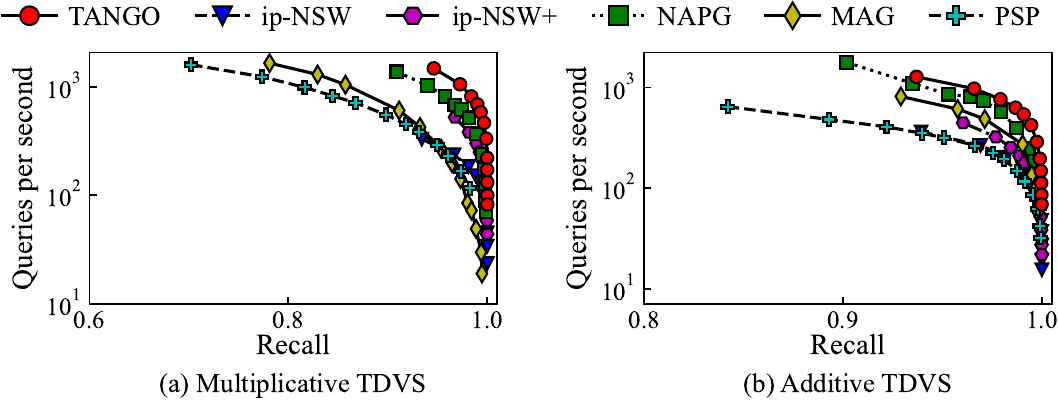}
  \caption{Comparison under the topic-independent timestamp distribution on OpenAI-1536.}
  \label{fig:time-distribution}
  \vspace*{-0.2cm}
\end{figure}

\subsubsection{Impact of Timestamp Distribution}
Figure~\ref{fig:time-distribution} evaluates the topic-independent setting, complementing the topic-correlated setting used in Figures~\ref{fig:query_multiplicative} and~\ref{fig:query_additive}.
In the default topic-correlated workload, vectors within each semantic cluster have timestamps sampled around a shared activity center $\mu_c$. Here, we remove this semantic--temporal correlation by sampling each timestamp independently from $\mathrm{Uniform}(0,H)$, regardless of cluster membership (i.e., semantic topics). This weakens the alignment between semantic and temporal neighborhoods and produces a different search geometry. 
TANGO achieves the best query performance under both TDVS modes. At recall $0.99$, it is $2.76\times$ and $2.02\times$ faster than the best-performing competitor under multiplicative and additive TDVS, respectively.
This demonstrates TANGO's robustness across both topic-independent and topic-correlated temporal workloads.

\begin{figure*}[t]
\vspace*{-1.1cm}
  \centering
  \includegraphics[width=\textwidth]
  {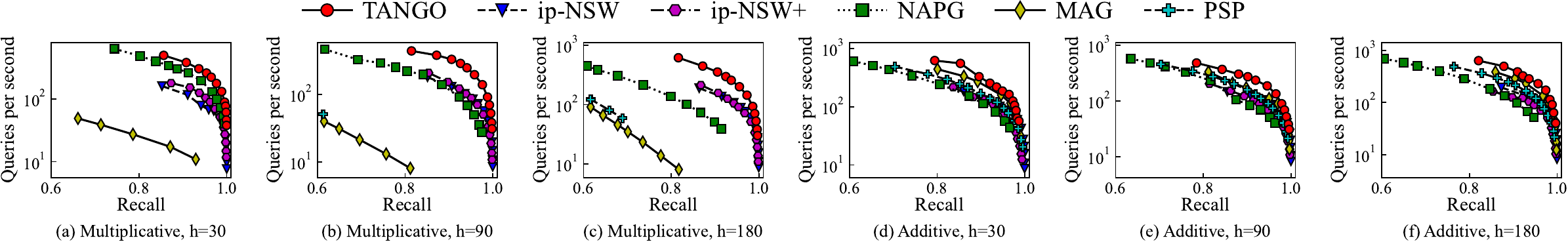}
  \caption{Comparison under different half-lives on OpenAI-3072.}
  \label{fig:half-life}
\end{figure*}

\begin{figure*}[t]
  \centering

  \begin{minipage}[t]{0.58\textwidth}
    \centering
    \includegraphics[width=\linewidth]
    {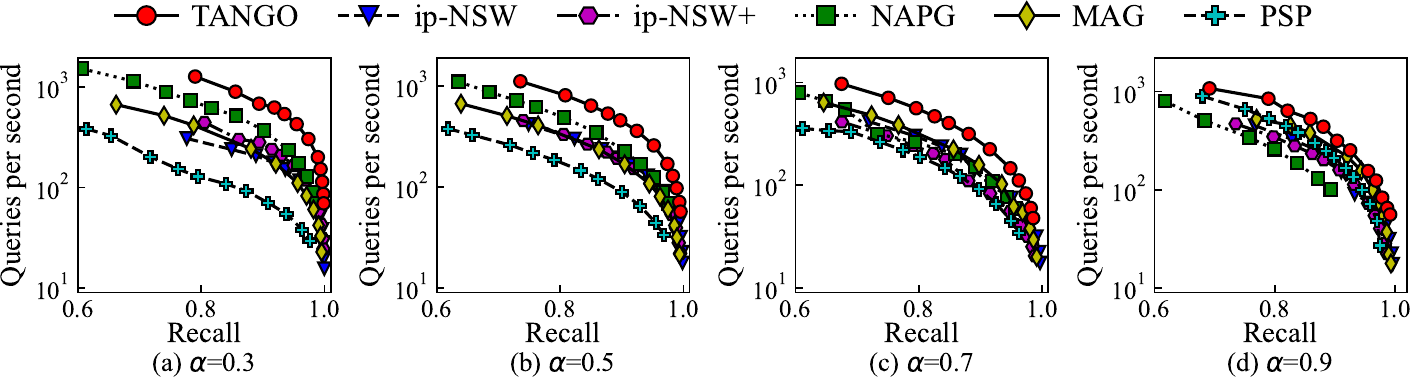}
    \captionof{figure}{Comparison under different semantic--temporal weights for additive TDVS on HotpotQA.}
    \label{fig:additive-alpha}
  \end{minipage}
  \hspace{0.025\textwidth}
  \begin{minipage}[t]{0.37\textwidth}
    \centering
    \includegraphics[width=\linewidth]
    {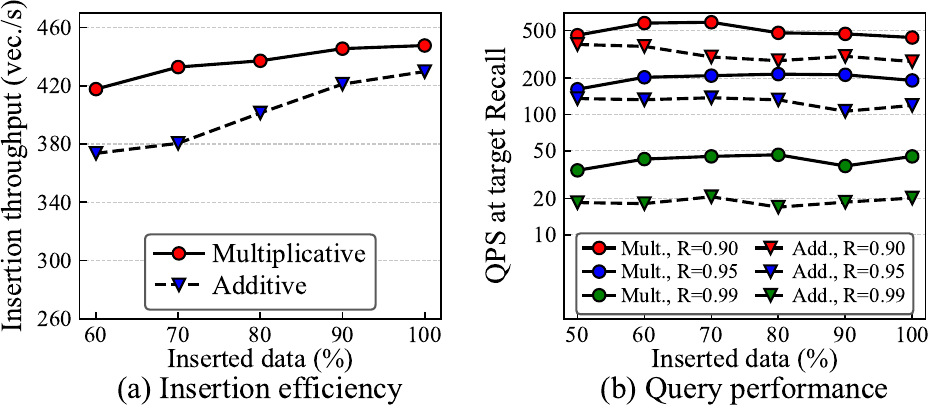}
    \captionof{figure}{Online insertion performance of TANGO on MSMARCO-10M.}
    \label{fig:dynamic-insertion}
  \end{minipage}

\end{figure*}

% \begin{figure*}[t]
%   \centering
%   \includegraphics[width=0.85\textwidth]
%   {./figures/additive_alpha_k50.pdf}
%   \caption{Comparison under different semantic--temporal weights for additive TDVS on HotpotQA.}
%   \label{fig:additive-alpha}
% \end{figure*}

\subsubsection{Impact of Half-Life}
Figure~\ref{fig:half-life} varies the half-life over
$h\in\{30,90,180\}$. A smaller $h$ induces faster temporal decay and stronger temporal selectivity, whereas a larger $h$ makes the objective increasingly dominated by semantic similarity.
TANGO achieves the best query performance in all six settings.
At recall $0.95$, TANGO outperforms the best-performing competitor by $2.00$--$2.21\times$ under multiplicative TDVS and $1.41$--$1.70\times$ under additive TDVS. 
Moreover, several competitors fail to reach high recall under multiplicative TDVS, while TANGO remains effective across all half-lives.
These results demonstrate TANGO's robustness across different half-life settings under both TDVS modes.

\subsubsection{Impact of the Semantic--Temporal Weight}
Figure~\ref{fig:additive-alpha} varies
$\alpha\in\{0.3,0.5,0.7,0.9\}$ under additive TDVS.
A smaller $\alpha$ assigns more weight to temporal freshness, whereas a larger $\alpha$ emphasizes semantic similarity.
TANGO achieves the best query performance for every value of $\alpha$.
At recall $0.95$, TANGO achieves speedups of $1.48$--$2.43\times$ over the best-performing competitor.
The largest gain occurs at $\alpha=0.3$ and generally narrows as $\alpha$ increases, consistent with the objective becoming closer to static semantic search.
Nevertheless, TANGO remains $1.55\times$ faster even at
$\alpha=0.9$.
These results demonstrate TANGO's robustness across different semantic--temporal weights under additive TDVS.

\subsection{Online Insertion Performance}
\label{sec:insertion}

% \begin{figure}[t]
%   \centering
%   \includegraphics[width=\linewidth]
%   {./figures/dynamic_insertion.pdf}
%   \caption{Online insertion performance of TANGO on MSMARCO-10M.}
%   \label{fig:dynamic-insertion}
%   \vspace*{-0.3cm}
% \end{figure}

Figure~\ref{fig:dynamic-insertion} shows TANGO's online insertion performance on the MSMARCO-10M dataset under both TDVS modes.
We build the initial index on the oldest 50\% of the vectors and then insert the remaining 50\% with a single thread in five batches, each containing the next 10\% in timestamp order.
As shown in Figure~\ref{fig:dynamic-insertion}(a), insertion throughput remains stable and even improves slightly as the index grows, because later arrivals benefit from increasingly well-covered recent-time neighborhoods, facilitating efficient local graph updates.
Figure~\ref{fig:dynamic-insertion}(b) further shows that TANGO maintains stable query performance throughout index growth at target recalls of $0.90$, $0.95$, and $0.99$.
These results demonstrate that TANGO supports efficient online insertion while maintaining robust query performance throughout continuous index growth.

\section{Related Work}
\label{sec:related_work}

\noindent \textbf{Time-aware retrieval and vector search.}
Temporal information has been incorporated into information retrieval through document-evolution modeling, recency-aware ranking, and temporal query modeling~\cite{elsas2010leveraging,dong2010towards,metzler2009improving}. Recent work on dynamic question answering, temporal RAG, and agent memory emphasizes retrieving current evidence~\cite{kasai2023realtime,wu2024time,qian2024timer4,park2023generative}, while EvoWiki and HoH evaluate the effects of evolving or outdated knowledge~\cite{tang2025evowiki,ouyang2025hoh}. For time-aware vector search, existing approaches mainly follow two paradigms. 
\emph{Filter-based} methods use temporal predicates, time ranges, or a specified time to define a hard eligible set~\cite{gollapudi2023filtered,patel2024acorn,
engels2024approximate,wang2025timestamp}. 
\emph{Rerank-based} methods first retrieve candidates under a static semantic objective and then incorporate recency through score fusion or post-retrieval reranking~\cite{dong2010towards,qian2024timer4,openclaw_memory_search}. 
In contrast, TDVS keeps all items eligible and incorporates continuous temporal decay directly into the vector-search objective, providing a native search formulation for freshness-sensitive workloads in which relevance evolves continuously over time.

\noindent \textbf{Kernel and feature mapping.}
A kernel is a similarity function that can be interpreted as an inner product after mapping objects into a possibly high-dimensional feature
space, allowing complex relationships to be analyzed using standard geometric tools~\cite{aronszajn1950theory}. 
Since this feature space may be implicit or infinite-dimensional, methods such as random Fourier features approximate certain kernels with explicit finite-dimensional vectors~\cite{rahimi2007random}.
In vector search, kernelized LSH supports approximate nearest-neighbor search under implicit kernel similarities~\cite{kulis2009kernelized},
while MIPS reductions construct explicit query and data mappings that enable existing hashing or nearest-neighbor indexes to preserve inner-product rankings~\cite{shrivastava2014asymmetric,neyshabur2015symmetric}.
Chronos introduces a novel implicit kernel-based metricization of the complete TDVS objective. It employs the Laplacian temporal kernel to derive an exact anchor-free TDVS metric and evaluates the induced distances directly from the original embeddings and timestamps.
Unlike prior kernel-based vector search approaches, Chronos requires neither kernel approximation nor explicit high-dimensional transformed representations, thereby avoiding both approximation error and feature-expansion overhead while preserving an exact metric formulation for TDVS-native indexing.

\noindent \textbf{Graph-based vector indexes.}
Vector-search techniques span several major families, including locality-sensitive hashing~\cite{tian2023db,wei2024det,wei2026pdet}, vector quantization~\cite{jegou2010product,gao2024rabitq,
gao2025practical}, subspace collision
~\cite{wei2025subspace,wei2026taco}, and graph~\cite{azizi2025graph,malkov2018efficient,fu2019fast}.
Graph-based indexes are the most widely adopted in practice because they offer strong recall--efficiency trade-offs across diverse workloads~\cite{wang2023graph,azizi2025graph}.
Representative ANNS graph indexes include HNSW~\cite{malkov2018efficient}, NSG~\cite{fu2019fast}, and DiskANN~\cite{subramanya2019diskann}.
For MIPS graph indexes, ip-NSW~\cite{morozov2018non} directly constructs an inner-product graph, while ip-NSW+~\cite{liu2020understanding}, NAPG~\cite{tan2021norm}, MAG~\cite{chen2025stitching}, and PSP~\cite{chen2025maximum} improve navigation through angular, norm-aware, hybrid-metric, or query-scaled geometries. 
These methods are designed for a fixed semantic objective.
TANGO introduces a TDVS-native hierarchy that assigns query-equivalent TimeLift geometries to different graph layers, with traversal guided by the exact TDVS score, thereby supporting both temporal locality and long-range semantic navigation.

\section{Conclusions} \label{sec:conclusion}

In this paper, we formalized time-decayed vector search (TDVS), which jointly models semantic similarity and temporal freshness. We derived STR as an exact reduction that enables existing MIPS indexes to support TDVS, and further presented Chronos, a TDVS-native metric framework with Query-Orthogonal TimeLift for controllable semantic--temporal geometry. 

Building on Chronos, we developed TANGO, a hierarchical graph index that combines temporal locality with long-range semantic connectivity and supports efficient online insertion. Extensive experiments on seven real-world datasets demonstrated that TANGO consistently outperforms state-of-the-art graph-based methods in query and index performance and remains robust across diverse temporal settings.

% \begin{acks}

% \end{acks}

%\clearpage

\bibliographystyle{ACM-Reference-Format}
\bibliography{ref}

@String{Computing = "Computing" }

@String{Computer = "{IEEE} Computer" }

@inproceedings{wei2026virtuouscycleaipoweredvector,
	title={{The Virtuous Cycle: AI-Powered Vector Search and Vector Search-Augmented AI}},
	author={Jiuqi Wei and Quanqing Xu and Chuanhui Yang},
	booktitle={2026 IEEE 42nd International Conference on Data Engineering (ICDE)},
	year={2026},
	organization={IEEE}
}

@inproceedings{ouyang2025hoh,
  title={Hoh: A dynamic benchmark for evaluating the impact of outdated information on retrieval-augmented generation},
  author={Ouyang, Jie and Pan, Tingyue and Cheng, Mingyue and Yan, Ruiran and Luo, Yucong and Lin, Jiaying and Liu, Qi},
  booktitle={Proceedings of the 63rd Annual Meeting of the Association for Computational Linguistics (Volume 1: Long Papers)},
  pages={6036--6063},
  year={2025}
}

@inproceedings{park2023generative,
  title={Generative agents: Interactive simulacra of human behavior},
  author={Park, Joon Sung and O'Brien, Joseph and Cai, Carrie Jun and Morris, Meredith Ringel and Liang, Percy and Bernstein, Michael S},
  booktitle={Proceedings of the 36th annual acm symposium on user interface software and technology},
  pages={1--22},
  year={2023}
}

@inproceedings{DBLP:conf/wims/EchihabiZP20,
  author       = {Karima Echihabi and
                  Kostas Zoumpatianos and
                  Themis Palpanas},
  title        = {Scalable Machine Learning on High-Dimensional Vectors: From Data Series
                  to Deep Network Embeddings},
  booktitle    = {{10th International Conference on Web Intelligence,
                  Mining and Semantics (WIMS)}},
  pages        = {1--6},
  year         = {2020}
}

@article{hydra1,
  author       = {Karima Echihabi and
                  Kostas Zoumpatianos and
                  Themis Palpanas and
                  Houda Benbrahim},
  title        = {The Lernaean Hydra of Data Series Similarity Search: An Experimental
                  Evaluation of the State of the Art},
  journal      = {Proc. {VLDB} Endow.},
  volume       = {12},
  number       = {2},
  pages        = {112--127},
  year         = {2018},
  url          = {http://www.vldb.org/pvldb/vol12/p112-echihabi.pdf},
  doi          = {10.14778/3282495.3282498},
  bibsource    = {dblp computer science bibliography, https://dblp.org}
}

@article{hydra2,
  author       = {Karima Echihabi and
                  Kostas Zoumpatianos and
                  Themis Palpanas and
                  Houda Benbrahim},
  title        = {Return of the Lernaean Hydra: Experimental Evaluation of Data Series
                  Approximate Similarity Search},
  journal      = {Proc. {VLDB} Endow.},
  volume       = {13},
  number       = {3},
  pages        = {403--420},
  year         = {2019},
  url          = {http://www.vldb.org/pvldb/vol13/p403-echihabi.pdf},
  doi          = {10.14778/3368289.3368303},
  bibsource    = {dblp computer science bibliography, https://dblp.org}
}

@article{kasai2023realtime,
  title={Realtime qa: What's the answer right now?},
  author={Kasai, Jungo and Sakaguchi, Keisuke and Le Bras, Ronan and Asai, Akari and Yu, Xinyan and Radev, Dragomir and Smith, Noah A and Choi, Yejin and Inui, Kentaro and others},
  journal={Advances in neural information processing systems},
  volume={36},
  pages={49025--49043},
  year={2023}
}

@inproceedings{qian2024timer4,
  title={TimeR4: Time-aware retrieval-augmented large language models for temporal knowledge graph question answering},
  author={Qian, Xinying and Zhang, Ying and Zhao, Yu and Zhou, Baohang and Sui, Xuhui and Zhang, Li and Song, Kehui},
  booktitle={Proceedings of the 2024 conference on empirical methods in natural language processing},
  pages={6942--6952},
  year={2024}
}

@inproceedings{ryu2025news,
  title={Is This News Still Interesting to You?: Lifetime-aware Interest Matching for News Recommendation},
  author={Ryu, Seongeun and Ko, Yunyong and Kim, Sang-Wook},
  booktitle={Proceedings of the 34th ACM International Conference on Information and Knowledge Management},
  pages={2515--2524},
  year={2025}
}

@inproceedings{wu2024time,
  title={Time-sensitve retrieval-augmented generation for question answering},
  author={Wu, Feifan and Liu, Lingyuan and He, Wentao and Liu, Ziqi and Zhang, Zhiqiang and Wang, Haofen and Wang, Meng},
  booktitle={Proceedings of the 33rd ACM International Conference on Information and Knowledge Management},
  pages={2544--2553},
  year={2024}
}

@inproceedings{tang2025evowiki,
  title={Evowiki: Evaluating llms on evolving knowledge},
  author={Tang, Wei and Cao, Yixin and Deng, Yang and Ying, Jiahao and Wang, Bo and Yang, Yizhe and Zhao, Yuyue and Zhang, Qi and Huang, Xuan-Jing and Jiang, Yu-Gang and others},
  booktitle={Proceedings of the 63rd Annual Meeting of the Association for Computational Linguistics (Volume 1: Long Papers)},
  pages={948--964},
  year={2025}
}

@article{wang2023graph,
  title={Graph-and Tree-based Indexes for High-dimensional Vector Similarity Search: Analyses, Comparisons, and Future Directions.},
  author={Wang, Zeyu and Wang, Peng and Palpanas, Themis and Wang, Wei},
  journal={IEEE Data Eng. Bull.},
  volume={47},
  number={3},
  pages={3--21},
  year={2023}
}

@article{azizi2025graph,
  title={Graph-based vector search: An experimental evaluation of the state-of-the-art},
  author={Azizi, Ilias and Echihabi, Karima and Palpanas, Themis},
  journal={Proceedings of the ACM on Management of Data},
  volume={3},
  number={1},
  pages={1--31},
  year={2025},
  publisher={ACM New York, NY, USA}
}

@inproceedings{wang2025timestamp,
  title={Timestamp Approximate Nearest Neighbor Search over High-Dimensional Vector Data},
  author={Wang, Yuxiang and He, Ziyuan and Tong, Yongxin and Zhou, Zimu and Zhong, Yiman},
  booktitle={2025 IEEE 41st International Conference on Data Engineering (ICDE)},
  pages={3043--3055},
  year={2025},
  organization={IEEE}
}

@inproceedings{gollapudi2023filtered,
  title={Filtered-diskann: Graph algorithms for approximate nearest neighbor search with filters},
  author={Gollapudi, Siddharth and Karia, Neel and Sivashankar, Varun and Krishnaswamy, Ravishankar and Begwani, Nikit and Raz, Swapnil and Lin, Yiyong and Zhang, Yin and Mahapatro, Neelam and Srinivasan, Premkumar and others},
  booktitle={Proceedings of the ACM Web Conference 2023},
  pages={3406--3416},
  year={2023}
}

@inproceedings{dong2010towards,
  title={Towards recency ranking in web search},
  author={Dong, Anlei and Chang, Yi and Zheng, Zhaohui and Mishne, Gilad and Bai, Jing and Zhang, Ruiqiang and Buchner, Karolina and Liao, Ciya and Diaz, Fernando},
  booktitle={Proceedings of the third ACM international conference on Web search and data mining},
  pages={11--20},
  year={2010}
}

@misc{openclaw_memory_search,
  author       = {{OpenClaw}},
  title        = {Memory Search},
  year         = {2026},
  howpublished = {\url{https://docs.openclaw.ai/concepts/memory-search}},
  note         = {OpenClaw documentation, accessed 2026-04-16}
}

@inproceedings{hinneburg2000nearest,
	title={What is the nearest neighbor in high dimensional spaces?},
	author={Hinneburg, Alexander and Aggarwal, Charu C and Keim, Daniel A},
	booktitle={26th Internat. Conference on Very Large Databases},
	pages={506--515},
	year={2000}
}

@inproceedings{borodin1999lower,
  title={Lower bounds for high dimensional nearest neighbor search and related problems},
  author={Borodin, Allan and Ostrovsky, Rafail and Rabani, Yuval},
  booktitle={Proceedings of the thirty-first annual ACM symposium on Theory of computing},
  pages={312--321},
  year={1999}
}

@article{murre2015replication,
  title={Replication and analysis of Ebbinghaus’ forgetting curve},
  author={Murre, Jaap MJ and Dros, Joeri},
  journal={PloS one},
  volume={10},
  number={7},
  pages={e0120644},
  year={2015},
  publisher={Public Library of Science San Francisco, CA USA}
}

@inproceedings{settles2016trainable,
  title={A trainable spaced repetition model for language learning},
  author={Settles, Burr and Meeder, Brendan},
  booktitle={Proceedings of the 54th annual meeting of the association for computational linguistics (volume 1: long papers)},
  pages={1848--1858},
  year={2016}
}

@article{morozov2018non,
  title={Non-metric similarity graphs for maximum inner product search},
  author={Morozov, Stanislav and Babenko, Artem},
  journal={Advances in Neural Information Processing Systems},
  volume={31},
  year={2018}
}

@inproceedings{liu2020understanding,
  title={Understanding and improving proximity graph based maximum inner product search},
  author={Liu, Jie and Yan, Xiao and Dai, Xinyan and Li, Zhirong and Cheng, James and Yang, Ming-Chang},
  booktitle={Proceedings of the AAAI Conference on Artificial Intelligence},
  volume={34},
  number={01},
  pages={139--146},
  year={2020}
}

@inproceedings{tan2021norm,
  title={Norm adjusted proximity graph for fast inner product retrieval},
  author={Tan, Shulong and Xu, Zhaozhuo and Zhao, Weijie and Fei, Hongliang and Zhou, Zhixin and Li, Ping},
  booktitle={Proceedings of the 27th ACM SIGKDD Conference on Knowledge Discovery \& Data Mining},
  pages={1552--1560},
  year={2021}
}

@article{chen2025maximum,
  title={Maximum Inner Product is Query-Scaled Nearest Neighbor},
  author={Chen, Tingyang and Fu, Cong and Wang, Kun and Ke, Xiangyu and Gao, Yunjun and Zhou, Wenchao and Ni, Yabo and Zeng, Anxiang},
  journal={Proceedings of the VLDB Endowment},
  volume={18},
  number={6},
  pages={1770--1783},
  year={2025},
  publisher={VLDB Endowment}
}

@inproceedings{chen2025stitching,
  title={Stitching inner product and euclidean metrics for topology-aware maximum inner product search},
  author={Chen, Tingyang and Fu, Cong and Ke, Xiangyu and Gao, Yunjun and Ni, Yabo and Zeng, Anxiang},
  booktitle={Proceedings of the 48th International ACM SIGIR Conference on Research and Development in Information Retrieval},
  pages={2341--2350},
  year={2025}
}

@inproceedings{reimers2019sentence,
  title={Sentence-bert: Sentence embeddings using siamese bert-networks},
  author={Reimers, Nils and Gurevych, Iryna},
  booktitle={Proceedings of the 2019 conference on empirical methods in natural language processing and the 9th international joint conference on natural language processing (EMNLP-IJCNLP)},
  pages={3982--3992},
  year={2019}
}

@article{zhang2025qwen3,
  title={Qwen3 embedding: Advancing text embedding and reranking through foundation models},
  author={Zhang, Yanzhao and Li, Mingxin and Long, Dingkun and Zhang, Xin and Lin, Huan and Yang, Baosong and Xie, Pengjun and Yang, An and Liu, Dayiheng and Lin, Junyang and others},
  journal={arXiv preprint arXiv:2506.05176},
  year={2025}
}

@misc{openai_embeddings_documentation,
  author       = {OpenAI},
  title        = {OpenAI Embeddings Documentation},
  year         = {2024},
  howpublished = {\url{https://platform.openai.com/docs/guides/embeddings}},
  note         = {Accessed: 2026-07-25}
}

@article{malkov2018efficient,
  title={Efficient and robust approximate nearest neighbor search using hierarchical navigable small world graphs},
  author={Malkov, Yu A and Yashunin, Dmitry A},
  journal={IEEE transactions on pattern analysis and machine intelligence},
  volume={42},
  number={4},
  pages={824--836},
  year={2018},
  publisher={IEEE}
}

@inproceedings{elsas2010leveraging,
  title={Leveraging temporal dynamics of document content in relevance ranking},
  author={Elsas, Jonathan L and Dumais, Susan T},
  booktitle={Proceedings of the third ACM international conference on Web search and data mining},
  pages={1--10},
  year={2010}
}

@article{patel2024acorn,
  title={Acorn: Performant and predicate-agnostic search over vector embeddings and structured data},
  author={Patel, Liana and Kraft, Peter and Guestrin, Carlos and Zaharia, Matei},
  journal={Proceedings of the ACM on Management of Data},
  volume={2},
  number={3},
  pages={1--27},
  year={2024},
  publisher={ACM New York, NY, USA}
}

@inproceedings{engels2024approximate,
  title={Approximate nearest neighbor search with window filters},
  author={Engels, Joshua and Landrum, Benjamin and Yu, Shangdi and Dhulipala, Laxman and Shun, Julian},
  booktitle={Proceedings of the 41st International Conference on Machine Learning},
  pages={12469--12490},
  year={2024}
}

@inproceedings{metzler2009improving,
  title={Improving search relevance for implicitly temporal queries},
  author={Metzler, Donald and Jones, Rosie and Peng, Fuchun and Zhang, Ruiqiang},
  booktitle={Proceedings of the 32nd international ACM SIGIR conference on Research and development in information retrieval},
  pages={700--701},
  year={2009}
}

@article{aronszajn1950theory,
  title={Theory of reproducing kernels},
  author={Aronszajn, Nachman},
  journal={Transactions of the American mathematical society},
  volume={68},
  number={3},
  pages={337--404},
  year={1950}
}

@article{rahimi2007random,
  title={Random features for large-scale kernel machines},
  author={Rahimi, Ali and Recht, Benjamin},
  journal={Advances in neural information processing systems},
  volume={20},
  year={2007}
}

@inproceedings{kulis2009kernelized,
  title={Kernelized locality-sensitive hashing for scalable image search},
  author={Kulis, Brian and Grauman, Kristen},
  booktitle={2009 IEEE 12th international conference on computer vision},
  pages={2130--2137},
  year={2009},
  organization={IEEE}
}

@article{shrivastava2014asymmetric,
  title={Asymmetric LSH (ALSH) for sublinear time maximum inner product search (MIPS)},
  author={Shrivastava, Anshumali and Li, Ping},
  journal={Advances in neural information processing systems},
  volume={27},
  year={2014}
}

@inproceedings{neyshabur2015symmetric,
  title={On symmetric and asymmetric LSHs for inner product search},
  author={Neyshabur, Behnam and Srebro, Nathan},
  booktitle={Proceedings of the 32nd International Conference on International Conference on Machine Learning-Volume 37},
  pages={1926--1934},
  year={2015}
}

@article{wei2025subspace,
  title={Subspace collision: An efficient and accurate framework for high-dimensional approximate nearest neighbor search},
  author={Wei, Jiuqi and Lee, Xiaodong and Liao, Zhenyu and Palpanas, Themis and Peng, Botao},
  journal={Proceedings of the ACM on Management of Data},
  volume={3},
  number={1},
  pages={1--29},
  year={2025},
  publisher={ACM New York, NY, USA}
}

@article{wei2026taco,
  title={Taco: Data-adaptive and query-aware subspace collision for high-dimensional approximate nearest neighbor search},
  author={Wei, Jiuqi and Liao, Zhenyu and Han, Ruoyu and Xu, Quanqing and Yang, Chuanhui and Palpanas, Themis},
  journal={Proceedings of the ACM on Management of Data},
  volume={4},
  number={3 (SIGMOD},
  pages={1--28},
  year={2026},
  publisher={ACM New York, NY, USA}
}

@article{wei2024det,
  title={DET-LSH: A Locality-Sensitive Hashing Scheme with Dynamic Encoding Tree for Approximate Nearest Neighbor Search},
  author={Wei, Jiuqi and Peng, Botao and Lee, Xiaodong and Palpanas, Themis},
  journal={Proceedings of the VLDB Endowment},
  volume={17},
  number={9},
  pages={2241--2254},
  year={2024},
  publisher={VLDB Endowment}
}

@article{wei2026pdet,
  title={Pdet-lsh: Scalable in-memory indexing for high-dimensional approximate nearest neighbor search with quality guarantees},
  author={Wei, Jiuqi and Lee, Xiaodong and Peng, Botao and Xu, Quanqing and Yang, Chuanhui and Palpanas, Themis},
  journal={IEEE Transactions on Knowledge and Data Engineering},
  year={2026},
  publisher={IEEE}
}

@article{tian2023db,
  title={DB-LSH 2.0: Locality-sensitive hashing with query-based dynamic bucketing},
  author={Tian, Yao and Zhao, Xi and Zhou, Xiaofang},
  journal={IEEE Transactions on Knowledge and Data Engineering},
  volume={36},
  number={3},
  pages={1000--1015},
  year={2023},
  publisher={IEEE}
}

@article{jegou2010product,
  title={Product quantization for nearest neighbor search},
  author={Jegou, Herve and Douze, Matthijs and Schmid, Cordelia},
  journal={IEEE transactions on pattern analysis and machine intelligence},
  volume={33},
  number={1},
  pages={117--128},
  year={2010},
  publisher={IEEE}
}

@article{gao2024rabitq,
  title={Rabitq: Quantizing high-dimensional vectors with a theoretical error bound for approximate nearest neighbor search},
  author={Gao, Jianyang and Long, Cheng},
  journal={Proceedings of the ACM on Management of Data},
  volume={2},
  number={3},
  pages={1--27},
  year={2024},
  publisher={ACM New York, NY, USA}
}

@article{gao2025practical,
  title={Practical and asymptotically optimal quantization of high-dimensional vectors in euclidean space for approximate nearest neighbor search},
  author={Gao, Jianyang and Gou, Yutong and Xu, Yuexuan and Yang, Yongyi and Long, Cheng and Wong, Raymond Chi-Wing},
  journal={Proceedings of the ACM on Management of Data},
  volume={3},
  number={3},
  pages={1--26},
  year={2025},
  publisher={ACM New York, NY, USA}
}

@article{fu2019fast,
  title={Fast approximate nearest neighbor search with the navigating spreading-out graph},
  author={Fu, Cong and Xiang, Chao and Wang, Changxu and Cai, Deng},
  journal={Proceedings of the VLDB Endowment},
  volume={12},
  number={5},
  pages={461--474},
  year={2019},
  publisher={VLDB Endowment}
}

@inproceedings{subramanya2019diskann,
  title={Diskann: Fast accurate billion-point nearest neighbor search on a single node},
  author={Subramanya, Suhas Jayaram and Devvrit and Kadekodi, Rohan and Krishaswamy, Ravishankar and Simhadri, Harsha Vardhan},
  booktitle={Proceedings of the 33rd International Conference on Neural Information Processing Systems},
  pages={13766--13776},
  year={2019}
}

@misc{qdrant_search_relevance,
  author = {{Qdrant}},
  title  = {Search Relevance},
  year   = {2026},
  url    = {https://qdrant.tech/documentation/search/search-relevance/},
  note   = {Qdrant documentation, accessed 2026-08-20}
}

@misc{zhu2025wikipedia_qwen,
  author = {Henry Zhu},
  title  = {Vector Database Embeddings Dataset},
  year   = {2025},
  url    = {https://huggingface.co/datasets/maknee/wikipedia_qwen_8b},
  note   = {Hugging Face dataset, accessed 2026-08-24}
}

@misc{qdrant_openai1536,
  author = {{Qdrant}},
  title  = {{DBpedia Entities with OpenAI text-embedding-3-large, 1536 Dimensions, 1M}},
  year   = {2024},
  url    = {https://huggingface.co/datasets/Qdrant/dbpedia-entities-openai3-text-embedding-3-large-1536-1M},
  note   = {Hugging Face dataset, accessed 2026-08-24}
}

@misc{qdrant_openai3072,
  author = {{Qdrant}},
  title  = {{DBpedia Entities with OpenAI text-embedding-3-large, 3072 Dimensions, 1M}},
  year   = {2024},
  url    = {https://huggingface.co/datasets/Qdrant/dbpedia-entities-openai3-text-embedding-3-large-3072-1M},
  note   = {Hugging Face dataset, accessed 2026-08-24}
}

@misc{cohere_msmarco_v21,
  author = {{Cohere Labs}},
  title  = {{TREC-RAG 2024 Corpus (MSMARCO 2.1) -- Encoded with Cohere Embed English v3}},
  year   = {2024},
  url    = {https://huggingface.co/datasets/CohereLabs/msmarco-v2.1-embed-english-v3},
  note   = {Hugging Face dataset, accessed 2026-08-24}
}

@misc{cohere_beir_embed,
  author = {{Cohere Labs}},
  title  = {{BEIR Embeddings with Cohere embed-english-v3.0 Model}},
  year   = {2024},
  url    = {https://huggingface.co/datasets/CohereLabs/beir-embed-english-v3},
  note   = {Hugging Face dataset, accessed 2026-08-24}
}

@misc{bigann2021,
  author = {{Big ANN Benchmarks}},
  title  = {{NeurIPS 2021 Big ANN Benchmark}},
  year   = {2021},
  url    = {https://big-ann-benchmarks.com/neurips21.html},
  note   = {Accessed 2026-08-24}
}

@misc{ipnsw_code,
  author       = {Stanislav Morozov},
  title        = {{ip-NSW}},
  howpublished = {\url{https://github.com/stanis-morozov/ip-nsw}},
  note         = {GitHub repository, accessed 2026-08-24}
}

@misc{ipnswplus_code,
  author       = {Jie Liu},
  title        = {{ip-NSW+: GraphMIPS Implementation}},
  howpublished = {\url{https://github.com/Jerry-liujie/ip-nsw/tree/GraphMIPS}},
  note         = {GitHub repository, accessed 2026-08-24}
}

@misc{mag_code,
  author       = {{ZJU-DAILY}},
  title        = {{MAG}},
  howpublished = {\url{https://github.com/ZJU-DAILY/MAG}},
  note         = {GitHub repository, accessed 2026-08-24}
}

@misc{psp_code,
  author       = {{ZJU-DAILY}},
  title        = {{PSP}},
  howpublished = {\url{https://github.com/ZJU-DAILY/PSP}},
  note         = {GitHub repository, accessed 2026-08-24}
}

\end{document}